\documentclass[a4paper,12pt,reqno]{amsart}

\usepackage{a4}
\usepackage{verbatim}
\usepackage{todonotes}
\usepackage{color}

\usepackage{amssymb}
\usepackage{color}
\usepackage{graphicx}
\usepackage{floatflt}
\usepackage{float} 		

\usepackage{amsthm}
\usepackage{amscd}
\usepackage{hyperref}

\newtheorem{definition}{Definition}

\newtheorem{theorem}[definition]{Theorem}
\newtheorem{lemma}[definition]{Lemma}
\newtheorem{remark}[definition]{Remark}

\newtheorem{corollary}[definition]{Corollary}

\newtheorem{proposition}[definition]{Proposition}

\def\XXint#1#2#3{{\setbox0=\hbox{$#1{#2#3}{\int}$}
		\vcenter{\hbox{$#2#3$}}\kern-.5\wd0}}

\makeatletter
\@addtoreset{definition}{section}
\@addtoreset{equation}{section}
\makeatother

\newcommand{\G}{\mathcal{G}}

\newcommand{\I}{\mathcal{I}}

\DeclareMathOperator{\Res}{Res}

\allowdisplaybreaks[4]

\begin{document}

\title[A simple formula for the $x$-$y$ symplectic transformation in topological recursion at all genus]{A simple formula for the $x$-$y$ symplectic transformation in topological recursion}

\author{Alexander Hock}

\address{Mathematical Institute, University of Oxford, Andrew Wiles Building, Woodstock Road,
	OX2 6GG, Oxford, UK \\
	{\itshape E-mail address:} \normalfont  
	\texttt{alexander.hock@maths.ox.ac.uk}}

\begin{abstract}
	Let $W_{g,n}$ be the correlators computed by Topological Recursion for some given spectral curve $(x,y)$ and $W^\vee_{g,n}$ for $(y,x)$, where the role of $x,y$ is inverted. These two sets of correlators $W_{g,n}$ and $W^\vee_{g,n}$ are related by the $x$-$y$ symplectic transformation. Bychkov, Dunin-Barkowski, Kazarian and Shadrin computed a functional relation between two slightly different sets of correlators. Together with Alexandrov, they proved that their functional relation is indeed the $x$-$y$ symplectic transformation in Topological Recursion. This article provides a fairly simple formula directly between $W_{g,n}$ and $W^\vee_{g,n}$ which holds by their theorem for meromorphic $x$ and $y$ with simple and distinct ramification points. Due to the recent connection between free probability and fully simple vs ordinary maps, we conclude a simplified moment-cumulant relation for moments and higher order free cumulants.
\end{abstract}

\maketitle

\markboth{\hfill\textsc\shortauthors}{\textsc{A simple formula for the $x$-$y$ symplectic transformation in Topological Recursion}\hfill}


\section{Introduction}
Invented in 2007 by Chekhov, Eynard and Orantin \cite{Chekhov:2006vd,Eynard:2007kz}, the theory of Topological Recursion (TR) reaches more and more applications in mathematical physics and pure mathematics. TR is a universal recursive procedure to compute from a given initial data $(\Sigma,x,y,B)$, called the \textit{spectral curve}, a family of differential forms $\omega_{g,n}$ labelled by two integers $g,n\geq 0$. The application of TR ranges from enumerative geometry and   random matrix theory over string theory to knot theory. We refer to \cite{Eynard:2014zxa} for a short overview.

We will stick to the following setup: Let $\omega_{g,n}$ be a family of
meromorphic differentials on $n$ products of Riemann surfaces $\Sigma$.
These $\omega_{g,n}$ are labeled by the genus $g$ and the number $n$
of marked points of a compact complex curve. These objects occur as differential forms
on a complex curve $E(x,y)=0$, understood in parametric
representation $x(z)$ and $y(z)$. For simplicity, we will assume that the complex curve is of genus zero.

The initial  data is the \textit{spectral curve} $(\Sigma,x,y,B)$ consisting of ramified coverings
$x,y: \Sigma \to \Sigma_0$ of Riemann surfaces. 
Let $\omega_{0,1}(z)=y(z)dx(z)$ be a differential 1-form and the \emph{Bergman kernel}
$\omega_{0,2}(z_1,z_2)=B(z_1,z_2)=\frac{dz_1\,dz_2}{(z_1-z_2)^2}$. From the initial data,
TR constructs a family of meromorphic differentials $\omega_{g,n+1}(z_1,...,z_n,z)$ with
$2g+n-2\geq 0$ via the following universal formula
(in which we abbreviate $I=\{z_1,...,z_n\}$):´
\begin{align}
\label{TR}
& \omega_{g,n+1}(I,z)
\\
& =\sum_{\alpha_i}
\Res\displaylimits_{q\to \alpha_i}
K_i(z,q)\bigg(
\omega_{g-1,n+2,0}(I, q,\sigma_i(q))
+\hspace*{-1cm} \sum_{\substack{g_1+g_2=g\\ I_1\uplus I_2=I\\
		(g_1,I_1)\neq (0,\emptyset)\neq (g_2,I_2)}}
\hspace*{-1.1cm} \omega_{g_1,|I_1|+1}(I_1,q)
\omega_{g_2,|I_2|+1}(I_2,\sigma_i(q))\!\bigg)\;.
\nonumber
\end{align}
This construction proceeds recursively 
in minus the Euler characteristic $-\chi=2g+n-2$. Further, we need to define:
\begin{itemize}
	\item The sum over the \textit{ramification points} $\alpha_i$ of the ramified
	covering  $x:\Sigma\to \Sigma_0$, defined via $dx(\alpha_i)=0$.
	\item The \textit{local Galois involution} $\sigma_i\neq \mathrm{id}$
	defined via $x(q)=x(\sigma_i(q))$ near $\alpha_i$ with the fixed
	point $\alpha_i$. 
	\item The \textit{recursion kernel} $K_i(z,q)
	=\frac{\frac{1}{2}\int^{q'=q}_{q'=\sigma_i(q)}
		B(z,q')}{\omega_{0,1}(q)-\omega_{0,1}(\sigma_i(q))}$  constructed
	from the initial data. 
\end{itemize}
We also assume that $x$ and $y$ have just simple ramification points, $y$ is regular at the ramification points of $x$ and vice versa, and both have no coinciding ramification points. 
All $\omega_{g,n}$ are symmetric differential forms with poles just located at the ramification points of $x$. 

If we interchange the role of $x$ and $y$, which means the spectral curve is of the form $(\Sigma,y,x,B)$, the differential forms are denoted by $\omega^\vee_{g,n}$. For instance, we have $\omega^\vee_{0,1}(z)=x(z) dy(z)$. These are also symmetric differential forms with poles just located at the ramification points of $y$ for $2g+n-2>0$.

\subsection{Symplectic Transformation}
For any spectral curve $(\Sigma, x,y,B)$, the \textit{free energy} is defined by 
\begin{align}\label{freeenergy}
	\mathcal{F}^g=\frac{1}{2-2g}\sum_{\alpha_i}
	\Res\displaylimits_{q\to \alpha_i}\Phi_{0,1}(x(q)) \omega_{g,1}(q)
\end{align}
for $g>1$, where $\Phi_{0,1}(x(q))=\int_o^{x(q)}\omega_{0,1}(q')$, i.e. $d\Phi_{0,1}(x(z))=\omega_{0,1}(z)$. For $g\in\{0,1\}$, the free energies are a bit subtle (see \cite{Eynard:2007kz}).

It is conjectured that the free energies $\mathcal{F}^g$ are invariant under any symplectic transformation of $(x,y)$, i.e. any transformation which leaves $|dx\wedge dy|$ invariant. Due to this observation, the free energies $\mathcal{F}^g$ are also called \textit{symplectic invariants}.

We list some transformations
\begin{itemize}
	\item $(x,y)\to \big(\frac{ax+b}{cx+d},\frac{(cx+d)^2}{ad-bc}y\big)$, where $\begin{pmatrix}
	a & b \\
	c & d 
	\end{pmatrix}\in PSL_2(\mathbb{C})$
	\item $(x,y)\to \big(x,y+R(x)\big)$, where $R(x)$ is any rational function
	\item $(x,y)\to (y,x)$.
\end{itemize}
The first and the second transformations actually also leave all $\omega_{g,n}$ invariant (not just $\mathcal{F}^g$). To see this, take the recursion kernel $K_i(z,q)$ of \eqref{TR}, the only dependence on $(x,y)$ is in the denominator, which is of the form
\begin{align*}
	\omega_{0,1}(q)-\omega_{0,1}(\sigma(q))=(y(q)-y(\sigma(q))) dx(q).
\end{align*}
This is indeed invariant under the first two transformations, since  for any rational function we have $R(x(q))=R(x(\sigma(q)))$. Note also that the first transformation conserves the ramification points. 

However, the third transformation $(x,y)\to (y,x)$, which we call the $x$-$y$ \textit{symplectic transformation}, changes the $\omega_{g,n}$'s completely, but $\mathcal{F}^g$ is conjectured to be invariant. As defined before, we denote the differential forms after the third transformation $\omega_{g,n}^\vee$ to distinguish them from $\omega_{g,n}$. Therefore, it is of considerable interest to understand the transformation of $\omega_{g,n}$ under the $x$-$y$ symplectic transformation.

There are essentially two main examples were the $x$-$y$ symplectic transformation relates two different interesting families of differential forms related to combinatorial problems:
\begin{itemize}
	\item 2-Matrix model \cite{Chekhov:2006vd,Eynard:2007nq}: the differential forms $\omega_{g,n}$ are related to genus $g$ bicoloured maps with $n$ boundaries in the first colour, whereas $\omega_{g,n}^\vee$  are related to genus $g$ bicoloured maps with $n$ boundaries in the second colour
	\item fully simple vs ordinary maps \cite{Bychkov:2021hfh,Borot:2021eif}: the differential $\omega_{g,n}$  are related to genus $g$ maps with $n$ ordinary boundaries, whereas $\omega_{g,n}^\vee$  are related to genus $g$ maps with $n$ fully simple boundaries.
\end{itemize}
Due to the last observation that the $x$-$y$ symplectic transformation relates ordinary and fully simple maps, there is a tremendous connection to the theory of free probability \cite{Voiculescu1986AdditionOC,Collins2006SecondOF}. In short, ordinary maps correspond to the generating series of higher order moments, whereas fully simple maps to the generating series of higher order free cumulants. This means that the relation between $\omega_{g,n}$ and $\omega_{g,n}^\vee$ corresponds to a moment-cumulant relation in free probability.

\subsection{Main Theorem}
This paper gives an explicit functional relation between the differential forms $\omega_{g,n}$ and $\omega_{g,n}^\vee$ which are related by the $x$-$y$ symplectic transformation for a broad class of spectral curves. The functional relation is much more convenient in terms of $W_{g,n}$ and $W^\vee_{g,n}$ which are related to $\omega_{g,n}$ and $\omega_{g,n}^\vee$ by
\begin{align}\label{cor}
	W_{g,n}(x(z_1),....,x(z_n))dx(z_1)...dx(z_n):=&\omega_{g,n}(z_1,...,z_n)\\\label{cor2}
	W^\vee_{g,n}(y(z_1),....,y(z_n))dy(z_1)...dy(z_n):=&\omega^\vee_{g,n}(z_1,...,z_n).
\end{align}
Recently, Bychkov, Dunin-Barkowski, Kazarian and Shadrin \cite{JEP_2022__9__1121_0,Bychkov:2020yzy,Bychkov:2021hfh,Bychkov:2022wgw} invented technical tools coming from certain operators on bosonic Fock space to compute already a functional relation for some special spectral curves related to weighted double Hurwitz numbers. It was then conjectured in \cite[Conj. 3.13]{Borot:2021thu} that this functional relation holds for all spectral curves. Finally, this is proved in \cite{Alexandrov:2022ydc} by Bychkov et al. together with Alexandrov for all spectral curves with meromorphic $x$ and $y$ and distinct simple ramification points.

However, the first known functional relations were given in terms of
\begin{align}\label{Wx}
	W_{g,n}(x_1,....,x_n)\prod_{i=1}^nx_i\qquad \text{and}\qquad 
	W^\vee_{g,n}(y_1,....,y_n)\prod_{i=1}^ny_i,
\end{align}
rather then just $W_{g,n}$ and $W^\vee_{g,n}$. We recap the functional relation for the $x$-$y$ symplectic transformation in Theorem \ref{Thm:Borot} and simplify it to Theorem \ref{Thm:mainintro}. The first known functional relations are of very complicated form consisting of three formal nested power series. We show that the functional relation for $W_{g,n}$ and $W_{g,n}^\vee$ is much simpler and give the most canonical way to relate the correlators computed by TR and its $x$-$y$ symplectic transformed analog. 

With some lack of notation, we are writing $x=x(z(y))=x(y)$ understood as a formal power series of $x$ in $y$. The functional relation then reads: 
\begin{theorem}\label{Thm:mainintro}
	Let $W_n(x_1,...,x_n):=\sum_{g=0}^\infty \hbar^{2g+n-2}W_{g,n}(x_1,...,x_n)$,
	 $S(u)=\frac{e^{u/2}-e^{-u/2}}{u}$ and for $I=\{i_1,...,i_n\}$
	\begin{align*}
	\hat{c}(u_I,x_I):=&\bigg(\prod_{i\in I}\hbar u_i S(\hbar u_i\partial_{x_i})\bigg)\big( W_{n}(x_I)\bigg)
	\end{align*}
	and for $I=\{j,j\}$ the special case 
	\begin{align*}
	\hat{c}(u_{I},x_I):=(\hbar u_j S(\hbar u_j\partial_{x_j}))(\hbar u_j S(\hbar u_j\partial_{x}))\bigg(W_{2}(x_j,x)-\frac{1}{(x_j-x)^2}\bigg)\bigg\vert_{x=x_j}.
	\end{align*}
	Let further be 
	\begin{align*}
	\hat{O} (x_i(y_i)):=&\sum_{m\geq0} \big(-\partial_{y_i}\big)^m\big(-x'_i(y_i)\big)
	[u_i^m]\frac{\exp\bigg(\hbar u_iS(\hbar u_i\partial_{x_i(y_i)})W_1(x_i(y_i))-y_iu_i\bigg)}{\hbar u_i}.
	\end{align*}
	
	Then, we have the functional relation
	\begin{align}\label{MainThmEq}
	W^{\vee}_{g,n}(y_1,...,y_n) =[\hbar^{2g-2+n}]\sum_{\Gamma\in\G_n}\frac{1}{|\mathrm{Aut}(\Gamma)|}\prod_{i=1}^n\hat{O} (x_i)\prod_{I\in\I(\Gamma)}\hat{c}(u_I,x_I),
	\end{align}
	where the graphs $\mathcal{G}_n$ are defined in Definition \ref{def:graph}.
\end{theorem}\noindent
The functional relation in \eqref{MainThmEq} is a tremendous simplification of all previously known functional relations. It can, for instance, be used to compute explicit formulae for intersection numbers on the moduli space of complex curves via its Laplace transform \cite{Hock:2023qii}.

The paper is organised as follows, we prove the main Theorem in Sec.\ref{Sec.simpl} in two steps. First, we recall the functional relation of \cite{Borot:2021thu}, which is proved in \cite{Alexandrov:2022ydc} to be the $x$-$y$ symplectic transformation, and simplify their operator $\vec{O}$ in Sec.\ref{Sec.vecO}. In Sec.\ref{Sec.x}, we eliminate explicit factors of $x(y)$ and $y$ and conclude the Theorem. We continue in Sec.\ref{Sec.Ex} with some examples which underpins that the Theorem gives the most canonical functional relation. In the case of $(g,n)=(2,1)$, the functional relation is explicitly verified with the Airy spectral curve. The new simple formula leads in Sec.\ref{Sec.comb} to a further combinatorial interpretation taking into account the expansion of $W_n$ and $S$ in $\hbar$. We finish in Sec.\ref{Sec.free} with an application of the functional relation to higher order free probability. We conclude with the most natural higher order moment-cumulant functional relation in free probability.

\subsection*{Acknowledgement}
I would like to thank Olivier Marchal, Sergey Shadrin and Raimar Wulkenhaar for their comments, and Petr Dunin-Barkowski for some explanations on his work. I also want to thank the reviewers for their comments which improved the readability of the article. This work was supported through
the Walter-Benjamin fellowship\footnote{``Funded by
	the Deutsche Forschungsgemeinschaft (DFG, German Research
	Foundation) -- Project-ID 465029630}.

\section{Simplification of the functional relation}\label{Sec.simpl}
First derived in Theorem \cite[Thm. 4.10]{Bychkov:2021hfh}, the functional relation in terms of \eqref{Wx} was already applied in \cite[Thm. 3.4]{Borot:2021thu} to generalise the moment-cumulant formula in higher order free probability (even to higher genus). It was then conjectured in \cite[Conj. 3.13]{Borot:2021thu} and proved in \cite{Alexandrov:2022ydc} that the functional relation holds indeed for any spectral curve with meromorphic $x,y$ and simple distinct ramification points.

We will take the formula appearing in \cite{Borot:2021thu}, which was derived rigorously in \cite{JEP_2022__9__1121_0,Bychkov:2020yzy,Bychkov:2021hfh}. It is also worth mentioning that the later consideration is a special case of \cite[Thm. 4.10]{Bychkov:2021hfh} with $\psi(p)=\log(p)$. In the case of $g=0$, the later kind of simplification was already achieved in \cite{Hock:2022wer}. 

The graphs appearing in the functional relation are characterised in the following way:
\begin{definition}\label{def:graph}
	Let $\mathcal{G}_{n}$ be the set of connected bicoloured graphs $\Gamma$ with $\bigcirc$-vertices and  $\bullet$-vertices, where the number of $\bigcirc$-vertices is $n$. A graph $\Gamma$ satisfies the following conditions:
	\begin{itemize}
		\item[-] the  $\bigcirc$-vertices are labelled from $1,...,n$
		\item[-] edges are only connecting $\bullet$-vertices with $\bigcirc$-vertices
		\item[-] $\bullet$-vertices have valence $\geq 2$.
	\end{itemize}
	For a graph $\Gamma\in \mathcal{G}_{n}$, let $r_{i}(\Gamma)$ be the valence of the $i^{\text{th}}$ $\bigcirc$-vertex. 
	
	Let $I\subset \{1,...,n\}$ be the set associated to a $\bullet$-vertex, where $I$ is the set of labels of $\bigcirc$-vertices connected to this $\bullet$-vertex. Let $\mathcal{I}(\Gamma)$ be the set of all sets $I$ for a given graph $\Gamma\in \mathcal{G}_{n}$.
\end{definition}\noindent
The automorphism group $\mathrm{Aut}(\Gamma)$ consists of permutations of edges or $\bullet$-vertices which preserve the structure of $\Gamma$ considering the labels. A graph $\Gamma\in \mathcal{G}_n$ is up to automorphisms completely characterised by the set $\mathcal{I}(\Gamma)$.

\begin{figure}[h]
	\scalebox{1}{
\begingroup%
  \makeatletter%
  \providecommand\color[2][]{%
    \errmessage{(Inkscape) Color is used for the text in Inkscape, but the package 'color.sty' is not loaded}%
    \renewcommand\color[2][]{}%
  }%
  \providecommand\transparent[1]{%
    \errmessage{(Inkscape) Transparency is used (non-zero) for the text in Inkscape, but the package 'transparent.sty' is not loaded}%
    \renewcommand\transparent[1]{}%
  }%
  \providecommand\rotatebox[2]{#2}%
  \newcommand*\fsize{\dimexpr\f@size pt\relax}%
  \newcommand*\lineheight[1]{\fontsize{\fsize}{#1\fsize}\selectfont}%
  \ifx\svgwidth\undefined%
    \setlength{\unitlength}{99.05496259bp}%
    \ifx\svgscale\undefined%
      \relax%
    \else%
      \setlength{\unitlength}{\unitlength * \real{\svgscale}}%
    \fi%
  \else%
    \setlength{\unitlength}{\svgwidth}%
  \fi%
  \global\let\svgwidth\undefined%
  \global\let\svgscale\undefined%
  \makeatother%
  \begin{picture}(1,0.38191891)%
    \lineheight{1}%
    \setlength\tabcolsep{0pt}%
    \put(0,0){\includegraphics[width=\unitlength,page=1]{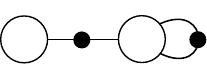}}%
    \put(0.11784883,0.14880473){\makebox(0,0)[t]{\lineheight{1.25}\smash{\begin{tabular}[t]{c}1\end{tabular}}}}%
    \put(0.69527614,0.15836527){\makebox(0,0)[t]{\lineheight{1.25}\smash{\begin{tabular}[t]{c}2\end{tabular}}}}%
    \put(0,0){\includegraphics[width=\unitlength,page=2]{graphEx.pdf}}%
  \end{picture}%
\endgroup%
}
	\caption{This is an example of a graph $\Gamma$ in $\mathcal{G}_2$ of Definition \ref{def:graph}. The set of labels associated to the $\bullet$-vertices is $\mathcal{I}(\Gamma)=\{(1,2),(1,2),(1,2),(2,2) \}$. The three $\bullet$-vertices connecting the two different $\bigcirc$-vertices can be permuted by automorphisms, as well as the two edges of the other $\bullet$-vertex. This gives in total $|\mathrm{Aut}(\Gamma)|=3!\cdot 2!=12$. }
	\label{fig:graphEx}
\end{figure}

Now, we can cite the Theorem (which is Conjecture 3.13 in \cite{Borot:2021thu} and Theorem 1.8 in \cite{Alexandrov:2022ydc}), which will be simplified (we adapted Theorem 3.4 together with Definition 3.3 of \cite{Borot:2021thu} to our notation\footnote{We have used the following identifications: $w_i\to x_i$, $X_i\to \frac{1}{y_i}$, $G_{|I|}^\vee (w_I)\to W_{|I|}(x_i)\prod_{i\in I}x_i$ and  $G_{|I|} (X_I)\to W^\vee_{|I|}(y_i)\prod_{i\in I}y_i$})

\begin{theorem}[\cite{Borot:2021thu}]\label{Thm:Borot}
	Let $S(x)=\frac{e^{x/2}-e^{-x/2}}{x}=\sum_{n=0}^\infty\frac{x^{2n}}{2^{2n}(2n+1)!}$,
	\begin{align}\label{cv}
		c(u_I,x_I):=\bigg(\prod_{i\in I}\hbar u_i S(\hbar u_ix_i\partial_{x_i})\bigg)\big( W_{|I|}(x_I)\prod_{i\in I}x_i\big)
	\end{align}
	together with the special case 
	for $I=\{j,j\}$
	\begin{align}\label{specialc}
		c(u_{I},x_I):=(\hbar u_j S(\hbar u_jx_j\partial_{x_j}))(\hbar u_j S(\hbar u_j x \partial_{x}))\bigg(W_{2}(x_j,x)x_jx-\frac{x_j x}{(x_j-x)^2}\bigg)\bigg\vert_{x=x_j}.
	\end{align}
	Let 
	\begin{align}\label{O-Operator}
		\vec{O} (x_i):=&\sum_{m\geq0} \big(-y_i\partial_{y_i}\big)^m\bigg(-\frac{y_ix'_i(y_i)}{x_i(y_i)}\bigg)\\\nonumber
		&\cdot [v^m_i]\sum_{r\geq 0}\bigg(\partial_p+\frac{v_i}{p}\bigg)^r\exp\bigg(v_i\frac{S(\hbar v_i\partial_p)}{S(\hbar \partial_p)}\log(p)-v_i\log(p)\bigg)\bigg\vert_{p=x_i(y_i)y_i}\\\nonumber
		&\cdot [u_i^r]\frac{\exp\bigg(\hbar u_iS(\hbar u_ix_i\partial_{x_i})(x_iW_1(x_i)-\frac{1}{\hbar})-u_i(x_i(y_i)y_i-1)\bigg)}{\hbar u_iS(\hbar u_i)}.
	\end{align}
Then we have for $2g-2+n>0$
\begin{align}\label{BorotThmEq}
	W^\vee_{g,n}(y_1,...,y_n) \prod_{i=1}^ny_i=\delta_{n,1}\Delta_g(x_1)+[\hbar^{2g-2+n}]\sum_{\Gamma\in\G_n}\frac{1}{|\mathrm{Aut}(\Gamma)|}\prod_{i=1}^n\vec{O} (x_i)\prod_{I\in\I(\Gamma)}c(u_I,x_I),
\end{align}
where 
\begin{align}\label{delta}
	\Delta_g(x)=&[\hbar^{2g}]\sum_{m\geq0}\big(-y\partial_{y}\big)^m[v^{m+1}]\exp\bigg(v\frac{S(\hbar v\partial_p)}{S(\hbar \partial_p)}\log(p)-v\log(p)\bigg)\bigg\vert_{p=x(y)y}\!\!\!\!\!\!\!\!\!\!\!\!
	\big(-y\partial_{y}\big)(x(y)y)
\end{align}
\end{theorem}
\begin{remark}
	Note that \eqref{specialc} differs from \cite[Def. 3.3.]{Borot:2021thu}, where it is probably just a misprint by copying it from \cite{JEP_2022__9__1121_0}. In \eqref{specialc}, we first take the action of $S$ and then take the limit $x\to x_j$. Equally, we understand the definition \eqref{cv} of $c(u_I,x_I)$ if several variables coincide.
\end{remark}
The functional relation of this Theorem is given in terms of three different sets of nested formal power series (in $\hbar$, $v_i$ and $u_j$). The important observation of Theorem \ref{Thm:main} is that the functional relation of Theorem \ref{Thm:Borot} is expressed between correlators
\begin{align*}
	W^{\vee}_{g,n}(y_1,...,y_n) \prod_{i=1}^ny_i\qquad \text{and}\qquad  W_{m}(x_1,...,x_m)\prod_{i=1}^mx_i,
\end{align*}
rather than between 
\begin{align}\label{TRCor}
	W^{\vee}_{g,n}(y_1,...,y_n) \qquad \text{and}\qquad  W_{g,m}(x_1,...,x_m)
\end{align}
which are actually computed by TR.
We will provide the functional relation between the second two sets of correlators \eqref{TRCor} computed by TR in terms of one formal power series in $\hbar$. The result will be much more canonical in the sense that all explicit factors of $x_i$ and $y_i$ on the rhs of \eqref{BorotThmEq} cancel and it is much more efficient for explicit computations. The cancellation of possible poles coming from factors of $x_i$ and $y_i$ was not proven before.

The proof of Theorem \ref{Thm:main} proceeds in two steps, firstly we simplify the operator $\vec{O}$. Secondly, we commute all explicit factors of $x_i$ in the weight $c$ with the operator $\vec{O}$ after its simplification.

\subsection{Simplification of $\vec{O}$}\label{Sec.vecO}
The second line of the operator $\vec{O}$ can be written for $v\in \mathbb{N}$ (this was already used in \cite{Bychkov:2020yzy})
\begin{align*}
	&\exp\bigg(v\frac{S(\hbar v\partial_p)}{S(\hbar \partial_p)}\log(p)-v\log(p)\bigg)=	\exp\bigg[\bigg(\frac{e^{v\hbar\partial_p/2}-e^{-v\hbar\partial_p/2}}{e^{\hbar\partial_p/2}-e^{-\hbar\partial_p/2}}-v\bigg)\log(p)\bigg]\\
	=&\exp\bigg[\bigg(\sum_{i=0}^{v-1}e^{(i+\frac{1-v}{2})\hbar \partial_p}-v\bigg)\log(p)\bigg]
	=\bigg(\frac{\hbar}{p}\bigg)^v\,\prod_{i=0}^{v-1}\bigg(i+\frac{1-v}{2}+\frac{p}{\hbar}\bigg)\\
	=&\bigg(\frac{\hbar}{p}\bigg)^v\frac{\Gamma(\frac{1+v}{2}+\frac{p}{\hbar})}{\Gamma(\frac{1-v}{2}+\frac{p}{\hbar})},
\end{align*}
where in the third step $e^{a\partial_p}f(p)=f(a+p)$ is used as a formal derivative.
For the last expression in terms of the $\Gamma$-function the assumption that $v\in \mathbb{N}$ can be relaxed to  $v\in\mathbb{C}$ due to the uniqueness Theorem of Bohr-Mollerup. 
The asymptotic expansion of the ratio of two $\Gamma$-functions is well-known \cite{pjm/1102613160,WOS:000368300500003} 
\begin{align}\label{GammaExpansion}
	\frac{\Gamma(z+t)}{\Gamma(z+s)}\sim z^{t-s}\sum_{n\geq0}\frac{(-1)^nB_n^{(t-s+1)}(t) (s-t)_n}{n!}z^{-n}\qquad \text{as $z\to \infty$,}
\end{align}
where $(t)_n=t (t+1)..(t+n-1)$ is the Pochhammer symbol and $B^{(a)}_n(t)$ is called the generalised Bernoulli polynomial defined as the generating function of
\begin{align*}
	\frac{x^ae^{tx}}{(e^x-1)^a}=\sum_{n=0}^\infty B_n^{(a)}(t)\frac{x^n}{n!}.
\end{align*}
The asymptotic expansion of the ratio of $\Gamma$-functions can be used, because the asymptotic behaviour $z\to \infty$ corresponds later to $\hbar\to 0$. One can therefore identify the formal expansion of the left hand side of \eqref{GammaExpansion} with that of the right hand side\footnote{Note that here $f(x)\sim \sum_n a_n x^n$ means that $|f(x)-\sum_{n=0}^N a_n x^n|\in \mathcal{O}(x^{N+1})$}.

Inserting the asymptotic expansion together with the definition of the generalised Bernoulli polynomials yields (with $z=\frac{p}{\hbar}$, $t=\frac{1+v}{2}$ and $s=\frac{1-v}{2}$ in \eqref{GammaExpansion})
\begin{align}\nonumber
	&\exp\bigg(v\frac{S(\hbar v\partial_p)}{S(\hbar \partial_p)}\log(p)-v\log(p)\bigg)=\bigg(\frac{\hbar}{p}\bigg)^v\frac{\Gamma(\frac{1+v}{2}+\frac{p}{\hbar})}{\Gamma(\frac{1-v}{2}+\frac{p}{\hbar})}\\\nonumber
	=&\sum_{n\geq0} v(v-1)...(v-n+1)\bigg(\frac{\hbar}{p}\bigg)^n[\hbar^n]\frac{\hbar^{v+1}}{(e^{\hbar/2}-e^{-\hbar/2})^{v+1}}\\\label{expexpansion}
	=&\sum_{n\geq0} v(v-1)...(v-n+1)\bigg(\frac{\hbar}{p}\bigg)^n[\hbar^n]\frac{1}{S(\hbar)^{v+1}}.
\end{align}
It remains the action of $\bigg(\partial_p+\frac{v_i}{p}\bigg)^r$ in the second line of \eqref{O-Operator} which becomes now fairly easy so that the entire second line of \eqref{O-Operator} simplifies to 
\begin{align}\nonumber
	&\sum_{r\geq 0}\bigg(\partial_p+\frac{v}{p}\bigg)^r\exp\bigg(v\frac{S(\hbar v\partial_p)}{S(\hbar \partial_p)}\log(p)-v\log(p)\bigg)\bigg\vert_{p=x(y)y}R_r\\\nonumber
	=&\sum_{r\geq 0}\bigg(\partial_p+\frac{v}{p}\bigg)^r\sum_{n=0}^\infty v(v-1)...(v-n+1)\bigg(\frac{\hbar}{p}\bigg)^n\bigg\vert_{p=x(y)y}\bigg([\hbar^n]\frac{1}{S(\hbar)^{v+1}}\bigg)R_r\\\label{Rr}
	=&\sum_{r,n\geq 0}v(v-1)...(v-(n+r)+1)\frac{1}{(x(y)y)^{n+r}}\bigg([\tilde{u}^{n}]\frac{1}{S(\hbar\tilde{u})^{v+1}}\bigg)R_r,
\end{align}
where $R_r$ is some rest depending on $r$. This rest was in \eqref{O-Operator} of the form $R_r=[u^r]\tilde{R}(u)$ for $r\geq 0$. The formal expansion in $\tilde{u}$ and $u$ of $R_r$ will be now combined in one formal expansion. 

In the special case of $n=1$ in Theorem \ref{Thm:Borot}, combining the formal expansion of $\tilde{u}$ and $u$ needs some special treatment. For this, we define
\begin{align*}
	\delta=\begin{cases}
		1, \quad \text{if $n=1$ in \eqref{BorotThmEq} and $\Gamma=\bigcirc\in \mathcal{G}_1$}\\
		0,\quad \text{else},
	\end{cases}
\end{align*}
where $\bigcirc\in \mathcal{G}_1$ is the unique graph consisting of just one $\bigcirc$-vertex and no edges.
Now, note that the formal $[u_1^r]$-expansion of the last line of \eqref{O-Operator} in the special case of $n=1$ in \eqref{BorotThmEq} and $\Gamma=\bigcirc$ is of the form $\frac{1}{u_i \hbar}$, because the weights $c(u_I,x_I)=1$ for the unique graph $\Gamma=\bigcirc\in \mathcal{G}_1$. However, this special term $\frac{1}{u_1\hbar}$ does effectively not contribute in Theorem \ref{BorotThmEq}, since the sum over $r$ starts with 0. 

Considering this, the combination of the formal $\tilde{u}$ and $u$-expansion in \eqref{Rr} can be performed uniformly . We rearrange the sum over $n$ and $r$ into a single sum over $k=n+r$ by identifying $\tilde{u}$ with $u$, where we have to take $\frac{\delta}{u\hbar }$ into account via
\begin{align*}
	&\sum_{r,n\geq 0}v(v-1)...(v-(n+r)+1)\frac{1}{(x(y)y)^{n+r}}[\tilde{u}^{n}]\frac{1}{S(\hbar\tilde{u})^{v+1}}[u^r]\tilde{R}(u)\\
	=&\sum_{r,n\geq 0}v(v-1)...(v-(n+r)+1)\frac{1}{(x(y)y)^{n+r}}[\tilde{u}^{n}]\frac{1}{S(\hbar\tilde{u})^{v+1}}[u^r]\bigg(\tilde{R}(u)-\frac{\delta}{u\hbar }\bigg)\\
	=&\sum_{k\geq 0}v(v-1)...(v-k+1)\frac{1}{(x(y)y)^{k}}[u^{k}]\frac{\tilde{R}(u)}{S(\hbar u)^{v+1}}\\
	&- \sum_{k\geq 0}v(v-1)...(v-k+1)\frac{1}{(x(y)y)^{k}}[u^{k+1}]\frac{\delta}{S(\hbar u)^{v+1}\hbar}.
\end{align*}

Finally, we commute all factors of $\frac{1}{y}$ out of the operator $\vec{O}$. In \eqref{O-Operator}, the factor $\bigg(-\frac{y_ix'_i(y_i)}{x_i(y_i)}\bigg)$ decreases $\frac{1}{y^{k}}$ to $\frac{1}{y^{k-1}}$. Then, writing for some formal power series $P(v)$ and some function $f(y)$
\begin{align*}
	&\sum_{m\geq0} \big(-y\partial_{y}\big)^m[v^m]P(v) f(y)
	=P(-y\partial_y)f(y)
\end{align*}
leads for any $k$ to
\begin{align}\label{V1}
	&\sum_{m\geq0} \big(-y\partial_{y}\big)^m[v^m]\frac{v (v-1)...(v-k+1)}{y^{k-1} S(\hbar u)^v} f(y)\\\nonumber
	=&(-y\partial_y)(-y\partial_y-1)...(-y\partial_y-k+1)S(u\hbar)^{y\partial_y}\frac{f(y)}{y^{k-1}}\\\nonumber
	=&(-y\partial_y)(-y\partial_y-1)...(-y\partial_y-k+1)\frac{f(S(\hbar u)y)}{(S(\hbar)y)^{k-1}}\\\label{movey}
	=&y\cdot (-\partial_y)^k\frac{f(S(\hbar u)y)}{S(\hbar)^{k-1}}=y(-\partial_y)^kS(u\hbar)^{y\partial_y}\frac{f(y)}{S(\hbar)^{k-1}}.
\end{align}
Here, we have used the well-know formal derivative $a^{y\partial_y}$ on some smooth function $\tilde{f}(y)$ as 
\begin{align}\label{formaldif}
	a^{y\partial_y}\tilde{f}(y)=\tilde{f}(a y).
\end{align}
Summarising the upper steps, with 
\begin{align*}
	f(y)=\bigg(-\frac{x'(y)}{S(u\hbar)x(y)^{k+1} \hbar u}\bigg)\bigg[
	\frac{\exp\bigg(\hbar uS(\hbar ux\partial_{x})\bigg(x(y)W_1(x(y))-\frac{1}{\hbar}\bigg)-u(x(y)y-1)\bigg)}{S(u\hbar)}-\delta\bigg]
\end{align*}

the following representation of the operator $\vec{O}$ is concluded:
\begin{lemma}\label{lem:simpleO}
	The operator $\vec{O}$ of \eqref{O-Operator} can be written as
	\begin{align}\label{Osim}
		\vec{O} (x(y))=&y\sum_{k\geq0}(-\partial_y)^k[u^k]S(u\hbar)^{y\partial_y}\bigg(-\frac{x'(y)}{S(u\hbar)^{k+1}x(y)^{k+1} \hbar u}\bigg)\\\nonumber
		&\cdot
		\exp\bigg(\hbar uS(\hbar ux\partial_{x})x(y)W_1(x(y))-ux(y)y\bigg)\\\nonumber
		&-y\sum_{k\geq0}(-\partial_y)^k[u^{k}]S(u\hbar)^{y\partial_y}\bigg(-\frac{x'(y)}{S(u\hbar)^{k}x(y)^{k+1} \hbar u}\bigg)\delta,
	\end{align}
where $\delta$ means that the last term does only contribute in \eqref{BorotThmEq} of Theorem \ref{Thm:Borot} if $n=1$ and $\Gamma\in \mathcal{G}_1$ is the unique graph consisting of just one $\bigcirc$-vertex and no edges.
\end{lemma}
Lemma \ref{lem:simpleO} shows that all explicit factors of $y$ in $\vec{O} (x(y))$ vanish. Note that the global prefactor of $y$ cancels perfectly the factors of $y_i$ on the lhs of the functional relation \eqref{BorotThmEq} in Theorem \ref{Thm:Borot}.

Following similar lines, we get also another representation for $\Delta_g$: 
\begin{lemma}\label{lem:Delta}
	 $\Delta_g$ in \eqref{delta} can be written as
	 \begin{align*}
	 	\Delta_g(x)=y(-\partial_y)^{2g-1}[\hbar^{2g}]\frac{S(\hbar)^{y\partial_y}}{S(\hbar)^{2g-1}}\bigg(-\frac{x'(y)}{x(y)^{2g}}\bigg).
	 \end{align*}
 \begin{proof}
 	From the definition of $\Delta_g$ in \eqref{delta} and from \eqref{expexpansion}, we first get
 	\begin{align*}
 		&\Delta_g(x)\\
 		=&[\hbar^{2g}]\sum_{m\geq 0}(-y\partial_y)^m[v^{m+1}] \sum_{n\geq 1}v(v-1)...(v-n+1)\bigg(\frac{\hbar}{x(y)y}\bigg)^n[u^n]\frac{1}{S(u)^{v+1}} (-y\partial_y)(x(y)y).
 	\end{align*}
 The term $[\hbar^{2g}]$ selects the
 coefficient for $n=2g$ in the series in $n$. The result does not depend on $\hbar$ anymore. We
 then set $u=\hbar$ and take the $[\hbar^{2g}]$-coefficient of $\frac{1}{S(\hbar)^{v+1}}$. The term $[v^{m+1}]\frac{1}{v}$ can be replaced by $[v^m]$.
 Letting act $\frac{1}{S(\hbar)^{v+1}}=\frac{S(\hbar)^{y\partial_y}}{S(\hbar)}$ as in \eqref{V1} and\eqref{movey} as a formal power series in $\hbar$, we derive
 \begin{align*}
 	=[\hbar^{2g}]\sum_{m\geq 0}(-y\partial_y)^m[v^{m}] (v-1)...(v-2g+1)\frac{(-\partial_{S(\hbar)y})(x(S(\hbar)y) S(\hbar)y)}{x(S(\hbar)y)^{2g}y^{2g-1} S(\hbar)^{2g}}.
 \end{align*}
The factor $\frac{1}{y^{2g-1}}$ can be removed with the same considerations as in \eqref{movey}
\begin{align*}
	=[\hbar^{2g}](-\partial_y)^{2g-1} \frac{-S(\hbar)yx'(S(\hbar)y)-x(S(\hbar)y)}{x(S(\hbar)y)^{2g} S(\hbar)^{2g}}.
\end{align*}
Moving the one explicit factor of $y$ through all $2g-1$ derivatives yields finally
\begin{align*}
	=[\hbar^{2g}]\bigg\{&y(-\partial_y)^{2g-1} \frac{-x'(S(\hbar)y)}{x(S(\hbar)y)^{2g} S(\hbar)^{2g-1}} \\
	&+(2g-1)(-\partial_y)^{2g-2} \frac{x'(S(\hbar)y)}{x(S(\hbar)y)^{2g} S(\hbar)^{2g-1}}\\
	&+(-\partial_y)^{2g-1} \frac{-1}{x(S(\hbar)y)^{2g-1} S(\hbar)^{2g}}\bigg\}.
\end{align*}
The last two lines cancel exactly and we end up with the assertion after extracting the formal derivative $S(\hbar)^{y\partial_y}$ in the first line.
 \end{proof}
\end{lemma}

Now, the exceptional appearance of $\Delta_g$ becomes comprehensible, it  vanishes exactly the additional term $\sim \delta$ in $\vec{O}$ of Lemma \ref{lem:simpleO}, which is present for $n=1$ and for the unique graph $\Gamma\in \mathcal{G}_1$ which consists of just one $\bigcirc$-vertex and no edges. This comes from the fact that we have to take the $[\hbar^{2g-1}]$ coefficient of the $\delta$-term in \eqref{Osim}, which becomes
\begin{align*}
	&-[\hbar^{2g-1}]y\sum_{k\geq0}(-\partial_y)^k[u^{k}]S(u\hbar)^{y\partial_y}\bigg(-\frac{x'(y)}{S(u\hbar)^{k}x(y)^{k+1} \hbar u}\bigg)\\
	=&-[\hbar^{2g-1}]y(-\partial_y)^{2g-1}[u^{2g-1}]S(u\hbar)^{y\partial_y}\bigg(-\frac{x'(y)}{S(u\hbar)^{2g-1}x(y)^{2g} \hbar u}\bigg)\\
	=&-[\hbar^{2g}]y(-\partial_y)^{2g-1}\frac{S(\hbar)^{y\partial_y}}{S(\hbar)^{2g-1}}\bigg(-\frac{x'(y)}{x(y)^{2g}}\bigg)\\
	=&-\Delta_g(x),
\end{align*}
where we had to take $k=2g-1$ and merged $\hbar u\to \hbar$. 

As a corollary, we simplified the Theorem \ref{Thm:Borot} to
\begin{corollary}\label{cor:relation}
	The following functional relation holds
	\begin{align}\label{eq:cor-relation}
		W^{\vee}_{g,n}(y_1,...,y_n) =[\hbar^{2g-2+n}]\sum_{\Gamma\in\G_n}\frac{1}{|\mathrm{Aut}(G)|}\prod_{i=1}^n\tilde{O} (x_i)\prod_{I\in\I(\Gamma)}c(u_I,x_I),
	\end{align}
where the operator $\tilde{O}$ is given by
\begin{align}\label{tildeO}
	\tilde{O}(x)=&\sum_{k\geq0}(-\partial_y)^k[u^k]S(u\hbar)^{y\partial_y}\bigg(-\frac{x'(y)}{S(u\hbar)^{k+1}x(y)^{k+1} \hbar u}\bigg)\\\nonumber
	&\cdot
	\exp\bigg(\hbar uS(\hbar ux\partial_{x})\bigg(x(y)W_1(x(y))\bigg)-ux(y)y\bigg)
\end{align}
and $c(u_I,x_I)$ as in Theorem \ref{Thm:Borot}.
\begin{proof}
	Inserting Lemma \ref{lem:simpleO} and \ref{lem:Delta} into Theorem \ref{Thm:Borot}, where $\Delta_g$ vanishes in the special case of $n=1$ due to the previous considerations with a term coming from the unique graph $\Gamma\in \mathcal{G}_1$ which consists of just one $\bigcirc$-vertex and no edges. The prefactor $y_i$ in $\vec{O} (x(y_i))$ of Lemma \ref{lem:simpleO} cancels $\prod_{i=1}^ny_i$ on the lhs of \eqref{BorotThmEq}. Therefore, without the $\delta$-term in \eqref{Osim} and the $y$ prefactor, we identify exactly $\tilde{O}(x)$ from $\vec{O}(x)$.
\end{proof}
\end{corollary}

\subsection{Vanishing of the explicit $x(y)$-dependence}\label{Sec.x}
In the theory of topological recursion more precisely for quantum spectral curves \cite{Norbury:2015lcn}, the primitive of the correlators is of great importance and appears also naturally in our functional relation. We define
\begin{align}\label{defPhi}
	\Phi_n(x_1,...,x_n):=\int_{o}^{x_1} dx'_1 ...\int_o^{x_n} dx'_n W_{n}(x'_1,...,x'_n),
\end{align}
where $o$ is some arbitrary base point. This definition holds for all genera $g$ respectively, i.e. $\Phi_{g,n}(x_1,...,x_n):=\int_o^{x_1} dx'_1 ...\int_o^{x_n} dx'_n W_{g,n}(x'_1,...,x'_n)$. It is very convenient to use those in the following, since the $\Phi_n$ are already present in Theorem \ref{Thm:Borot}. To see this let $I=\{1,...,n\}$, we write \eqref{cv} with the formal derivative \eqref{formaldif}
\begin{align}\nonumber
	c(u_I,x_I)=&\bigg(\prod_{i\in I}\frac{e^{\hbar u_i x_i\partial {x_i}/2}-e^{-\hbar u_i x_i\partial {x_i}/2}}{x_i\partial_{x_i}}\bigg)\bigg( W_{n}(x_I)\prod_{i\in I}x_i\bigg)\\\nonumber
	=&\bigg(\prod_{i\in I}\frac{1}{x_i\partial_{x_i}}\bigg)\bigg(W_{n}(e^{\hbar u_1/2}x_1,..,e^{\hbar u_n/2}x_n)e^{\hbar u_1/2}x_1...e^{\hbar u_n/2}x_n\\\nonumber
	&-W_{n}(e^{-\hbar u_1/2}x_1,e^{\hbar u_2/2}x_2,..,e^{\hbar u_n/2}x_n)e^{-\hbar u_1/2}x_1e^{\hbar u_2/2}x_2...e^{\hbar u_n/2}x_n\pm ...\\\nonumber
	&+(-1)^nW_{n}(e^{-\hbar u_1/2}x_1,..,(e^{-\hbar u_n/2}x_n)e^{-\hbar u_1/2}x_1...e^{-\hbar u_n/2}x_n\bigg)\\\label{Phi1}
	=&\Phi_{n}(e^{\hbar u_1/2}x_1,...,e^{\hbar u_n/2}x_n)+(-1)^1\Phi_{n}(e^{-\hbar u_1/2}x_1,e^{\hbar u_2/2}x_2...,e^{\hbar u_n/2}x_n)\pm....\\\nonumber
	&+(-1)^n \Phi_{n}(e^{-\hbar u_1/2}x_1,...,e^{-\hbar u_n/2}x_n)
\end{align} 
where for formal power series $(\partial_x)^{-1}=\int dx$ holds. This leads to $( x\partial_x)^{-1}f(x)=( \partial_x)^{-1}x^{-1}f(x)=\int dx \frac{f(x)}{x}$, since $( x\partial_x)^{-1}(x\partial_x)f(x)=(x\partial_x)( x\partial_x)^{-1}f(x)=f(x)$ for a formal power series $f(x)$. Interestingly, all arguments of  $\Phi_n$ are deformed by the factor  $e^{\pm\hbar u/2}$ and the result is a sum over all possible choices of deformed arguments $e^{\pm\hbar u_i/2}x_i$, which are $2^n$ terms. The sign $(-1)^k$ is given by the number of negatively deformed arguments, i.e. $k$-times $e^{-\hbar u_i/2}x_i$.

Now, it is crucial for the later proof by induction to extract the $(g,n)=(0,1)$-sector inside all operators $\tilde{O}$. We write $W^{(g>0)}_1(x):=W_1(x)-\frac{1}{\hbar}W_{0,1}(x)\in \mathcal{O}(\hbar)$. The previous computation gives in case $(g,n)=(0,1)$
\begin{align*}
	u S(\hbar u \partial_x) W_{0,1}(x(y)) x(y)=\frac{\Phi_{0,1}(e^{\hbar u/2}x(y))-\Phi_{0,1}(e^{-\hbar u/2}x(y))}{\hbar}
\end{align*}

so that the second line of the operator $\tilde{O}$ in \eqref{tildeO} takes the form
\begin{align}\label{Wg1}
	&\exp\bigg(\hbar uS(\hbar ux\partial_{x})\bigg(x(y)W^{(g>0)}_1(x(y))\bigg)\bigg)\sum_{j\geq0}\frac{\bigg(\frac{\Phi_{0,1}(e^{\hbar u/2}x(y))-\Phi_{0,1}(e^{-\hbar u/2}x(y))}{\hbar}-u x(y) y \bigg)^j}{j!}\\\nonumber
	=&\exp\bigg(\hbar uS(\hbar ux\partial_{x})\bigg(x(y)W^{(g>0)}_1(x(y))\bigg)\bigg)\sum_{j\geq0}( u x(y))^j\frac{\bigg(\frac{\Phi_{0,1}(e^{\hbar u/2}x(y))-\Phi_{0,1}(e^{-\hbar u/2}x(y))}{\hbar u x(y) }- y \bigg)^j}{j!}.
\end{align}
We have extracted $u^j$ to make the latter expression homogeneous in $\hbar u$.
The expansion of the function in the brackets of power $j$ starts with $(\hbar u)^2$
\begin{align}\label{PhiOu2}
	\bigg(\frac{\Phi_{0,1}(e^{\hbar u/2}x(y))-\Phi_{0,1}(e^{-\hbar u/2}x(y))}{\hbar u x(y) }- y\bigg)\in \mathcal{O}((\hbar u)^2),
\end{align}
since it is an even function in $\hbar u$ and the constant term vanishes due to $(\Phi_{0,1})'(x(y))=W_{0,1}(x(y))=y$ by definition \eqref{defPhi}. 

If we now combine Corollary \ref{cor:relation} and the expansion for the sector $(g,n)=(0,1)$ in \eqref{eq:cor-relation} with $\tilde{O}$ (for simplification in just one variable $x_i(y_i)$), we get
\begin{align}\nonumber
	&\sum_{k,j\geq0}(-\partial_y)^k[u^{k+1}]S(u\hbar)^{y\partial_y}\bigg(\frac{x'(y)}{S(u\hbar)^{k+1}x(y)^{k+1} \hbar }\bigg)\\\nonumber
	& \cdot\frac{( u x(y))^j\bigg(\frac{\Phi_{0,1}(e^{\hbar u/2}x(y))-\Phi_{0,1}(e^{-\hbar u/2}x(y))}{\hbar u x(y) }- y \bigg)^j}{j!} F(\hbar, e^{\pm \hbar u/2}x(y))\\\nonumber
	=&\sum_{k,j\geq0}(-\partial_y)^k[u^{k-j+1}]\bigg(\frac{x'(S(u\hbar)y)}{S(u\hbar)^{k-j+1}x(S(u\hbar)y)^{k-j+1} \hbar }\bigg)\\\nonumber
	&\cdot\frac{\bigg(\frac{\Phi_{0,1}(e^{\hbar u/2}x(S(u\hbar)y))-\Phi_{0,1}(e^{-\hbar u/2}x(S(u\hbar)y))}{\hbar u x(S(u\hbar)y )S(u\hbar) }- y \bigg)^j}{j!} F(\hbar, e^{\pm \hbar u/2}x(S(u\hbar)y))\\\label{Fhu}
	=&\sum_{k\geq0}(-\partial_y)^k[u^{k+1}]\sum_{j\geq 0}(-\partial_y)^j\bigg(\frac{x'(S(u\hbar)y)}{S(u\hbar)^{k+1}x(S(u\hbar)y)^{k+1} \hbar }\bigg)\\\nonumber
	&\cdot\frac{\bigg(\frac{\Phi_{0,1}(e^{\hbar u/2}x(S(u\hbar)y))-\Phi_{0,1}(e^{-\hbar u/2}x(S(u\hbar)y))}{\hbar u x(S(u\hbar)y )S(u\hbar) }- y \bigg)^j}{j!} F(\hbar,e^{\pm \hbar u/2} x(S(u\hbar)y)),
\end{align}
where $F(\hbar, e^{\pm \hbar u/2}x(y))\in\mathcal{O}(u^1)$ is a function independent of $j$ and $k$. More precisely, we have identified 
\begin{align}\label{FDef}
	F(\hbar, e^{\pm \hbar u/2}x(y))=\prod_{i} e^{\hbar u_iS(\hbar u_ix_i\partial_{x_i})\big(x_i(y_i)W^{(g>0)}_1(x_i(y_i))\big)}\sum_{\Gamma\in\G_n}\prod_{I\in\I(\Gamma)}c(u_I,x_I)
\end{align}
which consists of correlators coming from $c(u_I,x_I)$ and $W^{(g>0)}_1(x)$-terms, which are due to \eqref{Phi1} just of the form $\Phi_{n}(e^{\hbar u/2}x(y),...)-\Phi_{n}(e^{-\hbar u/2}x(y),...)$ and therefore just depend on $\hbar$ and $e^{\pm\hbar u/2}x(y)$. We have used in the first step the action of the formal derivative \eqref{formaldif} of $S(u\hbar)^{y\partial_y}$. 
As explained in \eqref{PhiOu2}, the term to power $j$ has vanishing constant
term. Furthermore, the contribution of the other factor is of $\mathcal{O}(u^1)$ because $F(\hbar, e^{\pm \hbar u/2}x(S(u\hbar)y))\in \mathcal{O}(u^1)$ and $\frac{x'(S(u\hbar)y)}{S(u\hbar)^{k+1}x(S(u\hbar)y)^{k+1} \hbar }\in \mathcal{O}(u^0)$. Therefore, terms in the sum for $k<j-1$ give a vanishing contribution, and we can shift $k\to k+j$.

To get rid of some new explicit factor of $y$, which would appear due to the expansion of $x(S(\hbar u )y)$, we will apply the following lemma:
\begin{lemma}\label{lem:tech}
	Let $\Phi_{0,1}(x)=\int_o^x dx'\,y(x')$, $f(x)$ some smooth function and $S(u)=\frac{e^{u/2}-e^{-u/2}}{u}$. Then, the following simplification holds as a formal power series in $u^{2}$
	\begin{align*}
		&\cosh(u/2)S(u)\sum_{j\geq0}(-\partial_y)^j\bigg[x'(S(u)y)
		f\big( x(S(u)y)\big)\frac{\bigg(\frac{\Phi_{0,1}(e^{u/2} x(S(u) y))-\Phi_{0,1}(e^{-u/2} x(S(u) y))}{u S(u)x(S(u)y)}-y\bigg)^j}{j!}\bigg]\\
		=&\sum_{j\geq0}(-\partial_y)^j\bigg[x'(y)
		f\bigg(\frac{x(y)}{\cosh(u/2)}\bigg)\frac{\bigg(\frac{\Phi_{0,1}(\frac{e^{u/2} x( y)}{\cosh(u/2)})-\Phi_{0,1}(\frac{e^{-u/2} x( y)}{\cosh(u/2)})}{u S(u)\frac{x(y)}{\cosh(u/2)}}-y\bigg)^j}{j!}\bigg].
	\end{align*} 
	This is a simplification in the sense that the rhs does not contain any explicit powers of $y$, whereas the lhs does through the expansion of $x(S(u)y)$.
	\begin{proof}
		See Appendix \ref{AppA}.
	\end{proof}
\end{lemma}

Using Lemma \ref{lem:tech} in \eqref{Fhu} after substituting $u\hbar\to u$ gives
\begin{align}\label{ZwSchritt1}
	=&\sum_{k\geq0}(-\partial_y)^k\hbar^k[u^{k+1}]\sum_{j\geq 0}(-\partial_y)^j\bigg(\frac{\cosh(u/2)^{k}x'(y)}{S(u)^{k+2}x(y)^{k+1} }\bigg)\\\nonumber
	&\cdot\frac{\bigg(\frac{\Phi_{0,1}(\frac{e^{ u/2}x(y)}{\cosh(u/2)})-\Phi_{0,1}(\frac{e^{- u/2}x(y)}{\cosh(u/2)})}{u \frac{x(y)}{\cosh(u/2)}S(u) }- y \bigg)^j}{j!} F\bigg(\hbar,  e^{\pm u/2}\frac{x(y)}{\cosh(u/2)}\bigg).
\end{align}
This seems to make the computation just worse, but each of the $\Phi_n$ (also appearing in $F$) has now the argument
\begin{align*}
	\frac{e^{\pm u/2}x(y)}{\cosh(u/2)}=\frac{2x(y)e^{\pm u/2}}{e^{u/2}+e^{-u/2}}=x(y)\pm \frac{e^{u/2}-e^{-u/2}}{e^{u/2}+e^{-u/2}}x(y).
\end{align*}
Substituting $2\frac{e^{u/2}-e^{-u/2}}{e^{u/2}+e^{-u/2}}x(y)=u\frac{S(u)}{\cosh(u/2)}x(y)=z$, we can perfectly apply the Lagrange-Bürmann formula\footnote{\textbf{Lagrange-Bürmann inversion theorem}: Let $f(w),H(w)$ be smooth functions with $f(0)=0$ and $f(w)=\frac{w}{\phi(w)}$ with $\phi(0)\neq 0$. For $z=f(w)$, let $g(z)=w$ be its formal inverse. Then the following holds
		$[z^n]H(g(z))=\frac{1}{n}[w^{n-1}](H'(w)\phi(w)^n).$
	Since $[z^n]H(g(z))=\frac{1}{n}[z^{n-1}]\partial_z H(g(z))$, it holds equivalently
	\begin{align}\label{LBEq}
		[z^{n-1}](H'(g(z))g'(z))=[w^{n-1}](H'(w)\phi(w)^n).
\end{align}}
\eqref{LBEq} for any $j$ and $k$. We set $z=f(u)=u\frac{S(u)}{\cosh(u/2)}x(y)$ that gives $\phi(u)=\frac{\cosh(u/2)}{S(u)x(y)}$ and $u=g(z)=\log\big(\frac{x(y)+\frac{z}{2}}{x(y)-\frac{z}{2}}\big)$ and find for all $k,j$
\begin{align}\nonumber
	&[u^{k+1}]\bigg(\frac{\cosh(u/2)^{k}x'(y)}{S(u)^{k+2}x(y)^{k+1} }\bigg)
	\frac{\bigg(\frac{\Phi_{0,1}(\frac{e^{ u/2}x(y)}{\cosh(u/2)})-\Phi_{0,1}(\frac{e^{- u/2}x(y)}{\cosh(u/2)})}{u \frac{x(y)}{\cosh(u/2)}S(u) }- y \bigg)^j}{j!} F\bigg(\hbar, e^{\pm u/2}\frac{x(y)}{\cosh(u/2)}\bigg)\\\nonumber
	=&x'(y)[u^{k+1}]\bigg(\frac{\phi(u)^{k+2} x(y)}{\cosh(u/2)^2}\bigg)\frac{\bigg(\frac{\Phi_{0,1}(x(y)+\frac{f(u)}{2})-\Phi_{0,1}(x(y)-\frac{f(u)}{2})}{f(u) }- y \bigg)^j}{j!} F\bigg(\hbar,  x(y)\pm\frac{f(u)}{2}\bigg)\\\label{ZwSchritt}
	=&x'(y)[z^{k+1}]\frac{\bigg(\frac{\Phi_{0,1}(x(y)+\frac{z}{2})-\Phi_{0,1}(x(y)-\frac{z}{2})}{z }- y \bigg)^j}{j!} F\bigg(\hbar, x(y)\pm\frac{z}{2}\bigg),
\end{align}
where we have used $g'(z)=\frac{x(y)}{(x(y)+\frac{z}{2})(x(y)-\frac{z}{2})}=\frac{\cosh(u/2)^2}{x(y)}$.

Now, we are ready to conclude the main Theorem, the simplification of Theorem \ref{Thm:Borot}:
\begin{theorem}\label{Thm:main}
	Let $S(x)=\frac{e^{x/2}-e^{-x/2}}{x}=\sum_{n=0}^\infty\frac{x^{2n}}{2^{2n}(2n+1)!}$ and for $I=\{i_1,...,i_n\}$
	\begin{align}\label{ccv}
		\hat{c}(u_I,x_I):=&\bigg(\prod_{i\in I}\hbar u_i S(\hbar u_i\partial_{x_i})\bigg)\big( W_{n}(x_I)\big)
	\end{align}
	and for $I=\{j,j\}$ the special case 
	\begin{align}\label{special}
		\hat{c}(u_{I},x_I):=(\hbar u_j S(\hbar u_j\partial_{x_j}))(\hbar u_j S(\hbar u_j\partial_{x}))\bigg(W_{2}(x_j,x)-\frac{1}{(x_j-x)^2}\bigg)\bigg\vert_{x=x_j}.
	\end{align}
	Let further be 
	\begin{align}\label{OO-Operator}
		\hat{O} (x_i(y_i)):=&\sum_{m\geq0} \big(-\partial_{y_i}\big)^m\big(-x'_i(y_i)\big)
		 [u_i^m]\frac{\exp\bigg(\hbar u_iS(\hbar u_i\partial_{x_i(y_i)})W_1(x_i(y_i))-y_iu_i\bigg)}{\hbar u_i}\\\nonumber
		 =&\sum_{m\geq0} \big(-\partial_{y_i}\big)^m\big(-x'_i(y_i)\big)
		 [u_i^m]\frac{\exp\bigg(\Phi_1\big(x_i(y_i)+\frac{\hbar u_i}{2}\big)-\Phi_1\big(x_i(y_i)-\frac{\hbar u_i}{2}\big)-y_iu_i\bigg)}{\hbar u_i},
	\end{align}
	where $\Phi_1(x):=\int_{o}^{x} dx' W_{1}(x')$.
	
	Then, we have for $2g+n-2>0$ the functional relation
	\begin{align}\label{ThmEq}
		W^{\vee}_{g,n}(y_1,...,y_n) =[\hbar^{2g-2+n}]\sum_{\Gamma\in\G_n}\frac{1}{|\mathrm{Aut}(\Gamma)|}\prod_{i=1}^n\hat{O} (x_i)\prod_{I\in\I(\Gamma)}\hat{c}(u_I,x_I),
	\end{align}
	where the graphs $\mathcal{G}_n$ are defined in Definition \ref{def:graph}.
	\begin{proof}
		Now, we have to bring all bricks together to prove the main theorem. First,
		we insert the computation carried out in \eqref{ZwSchritt} with the Lagrange-Bürmann inversion into \eqref{ZwSchritt1} and substitute $z\to\hbar u_i$. Second, we invert the computational steps carried out in \eqref{Fhu} (for each variable $x_i(y_i)$)
		\begin{align*}
			&\sum_{k,j\geq0}(-\partial_{y_i})^{j+k}
			x_i'(y_i)[u_i^{k}]\frac{\bigg(\frac{\Phi_{0,1}(x_i(y_i)+\frac{\hbar u_i}{2})-\Phi_{0,1}(x_i(y_i)-\frac{\hbar u_i}{2})}{\hbar u_i }- y_i \bigg)^j}{j!\hbar u_i} F\bigg(\hbar, x_i(y_i)\pm\frac{\hbar u_i}{2}\bigg)\\
			=&\sum_{k,j\geq0}(-\partial_{y_i})^{j+k}
			x'_i(y_i)[u_i^{k+j}]\frac{\bigg(\frac{\Phi_{0,1}(x_i(y_i)+\frac{\hbar u_i}{2})-\Phi_{0,1}(x_i(y_i)-\frac{\hbar u_i}{2})}{\hbar  }- y_iu_i \bigg)^j}{j!\hbar u_i} F\bigg(\hbar, x_i(y_i)\pm\frac{\hbar u_i}{2}\bigg)\\
			=&\sum_{k\geq0}(-\partial_{y_i})^{k}
			x_i'(y_i)[u_i^{k}]\frac{\exp\bigg(\frac{\Phi_{0,1}(x_i(y_i)+\frac{\hbar u_i}{2})-\Phi_{0,1}(x_i(y_i)-\frac{\hbar u_i}{2})}{\hbar  }- y_iu_i \bigg)}{\hbar u_i} F\bigg(\hbar, x_i(y_i)\pm\frac{\hbar u_i}{2}\bigg).
		\end{align*}
		The series is shifted $k+j\to k$ such that the series over $j$ becomes the exponential function. From the definition \eqref{FDef} of $F(\hbar,..)$, we are adding the terms with correlators $W^{(g>1)}_1$ to the exponential.
		The separation of \eqref{Wg1} is undone, which gives us exactly 
		$\hat{O}(x_i(y_i))$ as defined in \eqref{OO-Operator} of the theorem.
		This is done for all $x_i(y_i)$. 
		
		The remaining part of $F(\hbar,...)$ via definition \eqref{FDef} yields the sum over all graphs of $\G_n$ where the weights are changed. The new form of $F\bigg(\hbar, x(y)\pm\frac{\hbar u}{2}\bigg)$ instead of $F(\hbar, e^{\pm \hbar u/2}x(y))$ replaces for the weights $S(\hbar u x\partial_x)$ in $c(u_I,x_I)$ in \eqref{cv} of Theorem \ref{Thm:Borot} by $S(\hbar u \partial_x)$ in $\hat{c}(u_I,x_I)$ in \eqref{ccv} as defined in Theorem \ref{Thm:main}. This replacement happens due to the two different formal derivatives $e^{a\partial_x}f(x)=f(x+a)$ appearing in $S(\hbar u\partial_x)$  and $e^{ax\partial_x}f(x)=f(xe^a)$ appearing in $S(\hbar ux\partial_x)$, where recall $S(x)=\frac{e^{x/2}-e^{-x/2}}{x}$.  
	\end{proof}
\end{theorem}


Equation  \eqref{ThmEq} can be extended to the two cases with $2g+n-2\leq 0$ by replacing in \eqref{ThmEq} the sum starting at $m=-1$, because 
\begin{align*}
	W^\vee_{0,1}(y)=(-\partial_y)^{-1}(-x'(y))=x(y)
\end{align*}
gives the expected result.

The great benefit of the Theorem is that it gives the most compact formula for the functional relation between $W^{\vee}_{g,n}(y_1,...,y_n)$ and $W_{g,n}(x_1,...,x_{n})$, which are separately computed by TR with spectral curve $(y,x)$ and $(x,y)$, respectively. The number of terms appearing in the functional relation \eqref{ThmEq} of Theorem \ref{Thm:main} is reduced tremendously in comparison to \eqref{BorotThmEq} of Theorem \ref{Thm:Borot}. 

Furthermore, representing the functional relation on the $z$-plane, which means substituting (with some lack of notation) $y_i=y_i(z_i)$ and $x_i(y_i)=x_i(z_i)$ shows immediately that $W_{g,n}(y_1(z_1),...,y_n(z_n))$ can not have any poles at the zeros of $x_i(z_i)$ neither at the zeros of $y_i(z_i)$, which would first be implied by the functional relation of Theorem \ref{Thm:Borot}. The cancellation of possible poles coming from factors of $x_i$ and $y_i$ was not proven before.
The derivative in the definition of $\hat{O}$ \eqref{OO-Operator} wrt to $y_i$ would just generate additional poles at the ramification points of $y(z)$, which is indeed needed.

\subsection{Examples}\label{Sec.Ex}
To show how the formula of Theorem \ref{Thm:main} works and how simple the functional relation of Theorem \ref{Thm:main} is, we want to give some examples.

First, we add some information about the different orders of $\hbar$ contributing from different graphs. The set of graphs $\G_n$ is an infinite set, but as stated already in \cite{Borot:2021thu} just a finite number of graphs contribute for a fixed genus. This finite set is restricted by the number of loops $L$ (first Betti number), which has to be smaller or equal to the genus $L\leq g$. From the weights $\hat{c}$ of \eqref{ccv}, the smallest order in $\hbar$ is $\prod_{I\in \mathcal{I}(\Gamma)}\hbar^{2|I|-2}$. This happens because at leading order we have $\hbar^{|I|}$ coming from $\prod_{i\in I}\hbar u_i S(\hbar u_i\partial_{x_i})$ and $\hbar^{|I|-2}$ from $W_{|I|}(x_I)=\sum_{g}\hbar^{2g+|I|-2}W_{g,|I|}(x_I)$ for each weight $\hat{c}(u_I,x_I)$. Note that for a connected graph $L=1+E-V$ holds, where $E$ is the number of edges and $V$ the number of vertices. The number of edges for a graph $\Gamma\in \G_n$ defined by Definition \ref{def:graph} is $E=\sum_{I\in \I(\Gamma)}|I|$ and the number of vertices $V=n+|\I|$. Therefore, the number of loops is equal to $L=1+\sum_{I\in \I(\Gamma)}|I|-n+|\I|=1+\sum_{I\in \I(\Gamma)}(|I|-1)-n$. Inserting this into the leading order in $\hbar$ for a given graph $\Gamma\in \G_n$, we get 
\begin{align*}
	\prod_{I\in \mathcal{I}(\Gamma)}\hbar^{2|I|-2}=\hbar^{2\sum_{I\in \I(\Gamma)}(|I|-1)}=\hbar^{2L-2+2n}.
\end{align*}
The denominator of \eqref{OO-Operator} divides by $\hbar^n$ which gives in total a leading order of $\hbar^{2L-2+n}$. This proves that for each loop we are getting an additional factor of $\hbar^2$. Furthermore, the genus is an upper bound for the number of possible loops we have to consider for the graphs $\Gamma\in \G_n$.

\subsubsection{$(g,n)=(0,n)$}
For the special case of genus $g=0$, the functional relation simplifies even further. We have to take the order $[\hbar^{n-2}]$ in \eqref{ThmEq}. 
Due to the previous discussion, we have just to consider trees (no loops). 
We find that the leading contribution is given by trees. The weights $\hat{c}$ with just the genus $g=0$ contribution already provide the factor $\hbar^{n-2}$. Note further that the term of $\hbar^{-1}$ in the expansion of $W_1$ cancels by the term $-y_iu_i$ in the exponential.

Thereforem, all formal power series $S(\hbar u \partial_{x})$ are just expanded as 1. The operator $\hat{O}$ is therefore given for genus $g=0$ by
\begin{align*}
	\hat{O}(x_i(y_i))=\sum_{m\geq0} \big(-\partial_{y_i}\big)^m\big(-x'_i(y_i)\big)
	[u_i^m]\frac{1}{\hbar u_i}
\end{align*}
acting from the left on the weights $\hat{c}(u_I,x_I)$. Note that also within $\hat{c}(u_I,x_I)$ we have to take $1$ for $S(\hbar u \partial_{x})$. So we get 
\begin{align*}
	W^\vee_{0,n}(y_1,...,y_n) =&[\hbar^{n-2}]\sum_{\Gamma\in\G_n}\frac{1}{|\mathrm{Aut}(\Gamma)|}\prod_{i=1}^n\bigg(\sum_{m\geq0} \big(-\partial_{y_i}\big)^m\big(-x'_i(y_i)\big)
	[u_i^m]\frac{1}{\hbar u_i}\bigg)\\
	&\times \prod_{I\in \mathcal{I}(\Gamma)} \bigg(\prod_{i\in I}(\hbar u_i)\bigg)W_{0,|I|}(x_I)\hbar^{|I|-2}\\
	=&\sum_{\Gamma\in\G_n}\frac{1}{|\mathrm{Aut}(\Gamma)|}\prod_{i=1}^n \big(-\partial_{y_i}\big)^{r_i(\Gamma)-1}\big(-x'_i(y_i)\big)\prod_{I\in \mathcal{I}(\Gamma)} W_{0,|I|}(x_I),
\end{align*}
where $r_i(\Gamma)$ is the valence of the $i^{\text{th}}$ $\bigcirc$-vertex and we have to consider just the trees in $\G_n$ as mentioned above. We have used that for any tree $\Gamma\in\G_n$ the following holds $\prod_{I\in \I(\Gamma)}\hbar^{2|I|-2}=\hbar^{2n-2}$ and $\prod_{I\in \I(\Gamma)}\prod_{i\in I}u_i=\prod_{i=1}^nu_i^{r_i(\Gamma)}$. This result coincides with the already known simplification for $g=0$ derived in \cite{Hock:2022wer}.

\subsubsection{$(g,n)=(1,1)$}
As discussed before, each loop of a graph gives an additional factor of $\hbar^2$.
For genus $g=1$, we have two different graphs, the graph consisting of just one $\bigcirc$ without any edges (tree), and the graph consisting of one $\bigcirc$ of valence two connected to one $\bullet$-vertex of valence two giving one loop. All other graphs in $\mathcal{G}_1$ have more than one loop.

For the first graph, we have $\hat{c}(u_I,x_I)=1$ and expand the exponential of $\hat{O}(x(y))$ in \eqref{OO-Operator}, where we have to take into account the linear term of the expansion of the exponential. This gives for the operator at order $[\hbar^1]$
\begin{align*}
	[\hbar^1]\hat{O} (x(y))=&[\hbar^1]\sum_{m\geq0} \big(-\partial_{y}\big)^m\big(-x'(y)\big)
	[u^m]\frac{\hbar uS(\hbar u\partial_{x(y)})W_{1}(x(y))-yu}{\hbar u}\\
	=&-x'(y)W_{1,1}(x(y))-\partial^2_y\bigg(x'(y)\frac{\partial^2_{x(y)}y}{24}\bigg)\\
	=&-x'(y)W_{1,1}(x(y))-\frac{1}{24}\partial^3_y\frac{1}{x'(y)},
\end{align*}
where we have used $S(\hbar u \partial_{x})=1+\frac{\hbar^2u^2 \partial_x^2}{24}+\mathcal{O}(\hbar^4)$ and $W_{1}(x(y))=\frac{W_{0,1}(x(y))}{\hbar}+\hbar W_{1,1}(x(y))+\mathcal{O}(\hbar^3)=\frac{y}{\hbar}+\hbar W_{1,1}(x(y))+\mathcal{O}(\hbar^3)$.

For the second graph, there are two automorphisms $|\mathrm{Aut}(\Gamma)|=2$. We have to expand again $S(\hbar u \partial_{x})=1+\mathcal{O}(\hbar^2)$ just to leading order and get 
\begin{align*}
	&[\hbar^1]\hat{O} (x(y))\hbar^2u^2\hat{W}_{0,2}(x(y),x(y))
	=\partial_y x'(y)\hat{W}_{0,2}(x(y),x(y)),
\end{align*}
where $\hat{W}_{0,2}(x_1,x_2):=W_{0,2}(x_1,x_2)-\frac{1}{(x_1-x_2)^2}$ is the regularised version, coming from the special case \eqref{special} with well-defined diagonal. Putting al together we have
\begin{align}\label{W11}
	W^\vee_{1,1}(y)=-x'(y)W_{1,1}(x(y))-\frac{1}{24}\partial^3_y\bigg(\frac{1}{x'(y)}\bigg) +\frac{1}{2}\partial_y \big[x'(y)\hat{W}_{0,2}(x(y),x(y))\big] ,
\end{align}
which is equivalent to \cite[Lemma C.1]{Eynard:2007kz} with trivial Bergman projective connection (since we have a genus $g=0$ spectral curve) or \cite[eq. (6.1)]{Hock:2022wer}. In \cite[eq. (28)]{Borot:2021thu}, this example was also computed from Theorem \ref{Thm:Borot} and consisted of seven terms, which were after tedious computation simplified to \eqref{W11}.

\subsubsection{$(g,n)=(1,2)$}\label{exgn12}
As explained above, at genus one a loop gives the additional factor of $\hbar^2$, which restricts us to all graphs in $\mathcal{G}_2$ with at most one loop. To compute the functional relation for the $(g,n)=(1,2)$ case, we have to consider six different graphs, which are up to permutations of the $\bigcirc$-vertices given in Fig.\ref{fig:gn12}.
\begin{figure}[h]
	\scalebox{1}{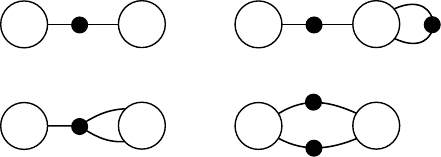}
	\caption{These are the graphs (up to possible permutation of the $\bigcirc$-vertices after label them), which contribute to $(g,n)=(1,2)$. The upper left graph has $|\mathrm{Aut}(\Gamma)|=1$, whereas all the other graphs have $|\mathrm{Aut}(\Gamma)|=2$.}
	\label{fig:gn12}
\end{figure}
The upper left graph with trivial automorphism gives due to the formula \eqref{ThmEq} the following contribution (where we have expanded already $O$ to the second order)
\begin{align}\nonumber
	[\hbar^2]&\sum_{m_1,m_2\geq 0}(-\partial_{y_1})^{m_1}(-\partial_{y_2})^{m_2}x'_1(y_1)x'_2(y_2)[u_1^{m_1}][u_2^{m_2}] \\\nonumber
	&\times\frac{1+u_1\hbar^2W_{1,1}(x_1(y_1)) +u_2\hbar^2W_{1,1}(x_2(y_2))+\frac{u_1^3\hbar^2}{24}\partial^2_{x_1(y_1)}y_1 +\frac{u_2^3\hbar^2}{24}\partial^2_{x_2(y_2)}y_2}{u_1u_2\hbar^2}  \\\nonumber
	&\times u_1u_2\hbar^2 \bigg(1+\frac{\hbar^2u_1^2\partial^2_{x_1(y_1)}+\hbar^2u_2^2\partial^2_{x_2(y_2)}}{24}\bigg) \big(W_{0,2}(x_1(y_1),x_2(y_2))+\hbar^2 W_{1,2}(x_1(y_1),x_2(y_2))\big)\\\nonumber
	=&x'_1(y_1)x'_2(y_2)W_{1,2}(x_1(y_1),x_2(y_2))\\\nonumber
	&-\partial_{y_1}\bigg(x'_1(y_1)x'_2(y_2)W_{1,1}(x_1(y_1))W_{0,2}(x_1(y_1),x_2(y_2))\bigg)\\\nonumber
	&-\partial_{y_2}\bigg(x'_1(y_1)x'_2(y_2)W_{1,1}(x_2(y_2))W_{0,2}(x_1(y_1),x_2(y_2))\bigg)\\\label{W121}
	&-\frac{1}{24}\partial^3_{y_1}\bigg(x'_1(y_1)x'_2(y_2) \big(\partial^2_{x_1(y_1)}y_1\big) W_{0,2}(x_1(y_1),x_2(y_2))\bigg)\\\nonumber
	&-\frac{1}{24}\partial^3_{y_2}\bigg(x'_1(y_1)x'_2(y_2) \big(\partial^2_{x_2(y_2)}y_2\big) W_{0,2}(x_1(y_1),x_2(y_2))\bigg)\\\nonumber
	&+\frac{1}{24}\partial^2_{y_1}\bigg(x'_1(y_1)x'_2(y_2) \partial^2_{x_1(y_1)} W_{0,2}(x_1(y_1),x_2(y_2))\bigg)\\\nonumber
	&+\frac{1}{24}\partial^2_{y_2}\bigg(x'_1(y_1)x'_2(y_2) \partial^2_{x_2(y_2)} W_{0,2}(x_1(y_1),x_2(y_2))\bigg).
\end{align}

For the other graphs from Fig.\ref{fig:gn12}, we just have to expand all $S(\hbar u \partial_{x})$ to the constant term´ 1.  The left lower graph in Fig.\ref{fig:gn12} yields therefore two terms:
\begin{align}\label{W122}
	&-\frac{1}{2}\partial_{y_1}\big( x'_1(y_1)x'_2(y_2)W_{0,3}(x_1(y_1),x_1(y_1),x_2(y_2))\big)\\\nonumber
	&-\frac{1}{2}\partial_{y_1}\big(x'_1(y_1)x'_2(y_2)W_{0,3}(x_1(y_1),x_2(y_2),x_2(y_2))\big);
\end{align}
the right upper graph in Fig.\ref{fig:gn12} also two terms:
\begin{align}
	\label{W123}
	&\frac{1}{2}\partial_{y_1}^2\big( x'_1(y_1)x'_2(y_2)W_{0,2}(x_1(y_1),x_2(y_2)) \hat{W}_{0,2}(x_1(y_1),x_1(y_1))\big)\\\nonumber
	&+\frac{1}{2}\partial_{y_2}^2\big( x'_1(y_1)x'_2(y_2)W_{0,2}(x_1(y_1),x_2(y_2)) \hat{W}_{0,2}(x_2(y_2),x_2(y_2))\big),
\end{align}
where $\hat{W}_{0,2}(x_1,x_2):=W_{0,2}(x_1,x_2)-\frac{1}{(x_1-x_2)^2}$; and the right lower graph in Fig.\ref{fig:gn12} just one term:
\begin{align}
	\label{W124}
	&\frac{1}{2}\partial_{y_1}\partial_{y_2}\big( x'_1(y_1)x'_2(y_2)W_{0,2}(x_1(y_1),x_2(y_2)) W_{0,2}(x_1(y_1),x_2(y_2))\big).
\end{align}

Adding all together, the functional relation reads
\begin{align*}
	W^\vee_{1,2}(y_1,y_2)=\eqref{W121}+\eqref{W122}+\eqref{W123}+\eqref{W124},
\end{align*}
which is the same as the already derived result in \cite[Prop. 6.1]{Hock:2022wer} through a loop insertion operator.

\subsubsection{$(g,n)=(2,1)$}\label{exgn21} The last example will be the simplest at genus $g=2$. Since each loop of a graph gives a factor of $\hbar^2$, we have to consider all graphs in $\mathcal{G}_1$ up to two loop. We have four graphs in $\mathcal{G}_1$ to consider at genus $g=2$, which are shown in Fig.\ref{fig:gn21}.
\begin{figure}[h]
	\scalebox{1}{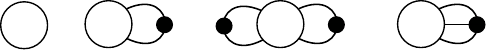}
	\caption{These graphs contribute to $(g,n)=(2,1)$. If we call the graphs $\Gamma_1,\Gamma_2,\Gamma_3,\Gamma_4$ from left to write, then we have $|\mathrm{Aut}(\Gamma_1)|=1$, $|\mathrm{Aut}(\Gamma_2)|=2$, $|\mathrm{Aut}(\Gamma_3)|=8$ and $|\mathrm{Aut}(\Gamma_4)|=6$}
	\label{fig:gn21}
\end{figure}
For the first graph $\Gamma_1$, we just have to expand the operator $\hat{O}$ defined in \eqref{OO-Operator}. This gives at order $[h^3]$
\begin{align}\nonumber
	&[\hbar^3]\hat{O}(x(y))=\sum_{m\geq 0}(-\partial_y)^m(-x'(y))[u^m]\bigg(W_{2,1}(x(y))+\frac{u^2}{24}\partial_{x(y)}^2W_{1,1}(x(y))\\\nonumber
	&\qquad \qquad+ \frac{u^4}{1920}\partial_{x(y)}^4y+\frac{u}{2}\big(W_{1,1}(x(y))\big)^2+\frac{u^3}{24}W_{1,1}(x(y))\partial_{x(y)}^2y+\frac{u^5}{24\cdot 24\cdot 2}\big(\partial_{x(y)}^2y\big)^2\bigg)\\\nonumber
	=&-x'(y)W_{2,1}(x(y))-\frac{1}{24}\partial^2_y\big(x'(y)\partial_{x(y)}^2W_{1,1}(x(y))\big)-\frac{1}{1920}\partial^4_y\big(x'(y)\partial_{x(y)}^4y\big)\\\label{W211}
	&+\frac{1}{2}\partial_y\big(x'(y)W_{1,1}(x(y))W_{1,1}(x(y))\big)+\frac{1}{24}\partial^3_y\big(x'(y)W_{1,1}(x(y))\partial_{x(y)}^2y\big)\\\nonumber
	&+\frac{1}{24\cdot 24\cdot 2}\partial^5_y\big(x'(y)(\partial_{x(y)}^2y)^2\big),
\end{align}
where we have used the expansion $S(\hbar u \partial_x)=1+\frac{\hbar^2u^2\partial_x^2}{24}+\frac{\hbar^4u^4\partial_x^4}{1920}+\mathcal{O}(\hbar^6)$.
For the second graph $\Gamma_2$ in Fig.\ref{fig:gn21} we have due to formula \eqref{ThmEq}
\begin{align}\nonumber
	&[\hbar^3]\frac{1}{2}\hat{O}(x(y))u^2\hbar^2 \bigg(1+2\frac{\hbar^2u^2\partial^2_{x(y)}}{24}\bigg) \big(\hat{W}_{0,2}(x(y),x(y))+\hbar^2 W_{1,2}(x(y),x(y))\big)\\\nonumber
	=&\frac{1}{2}\partial_y\big(x'(y) W_{1,2}(x(y),x(y))\big)-\frac{1}{2}\partial_y^2\big(x'(y) \hat{W}_{0,2}(x(y),x(y))W_{1,1}(x(y))\big)\\\label{W212}
	&-\frac{1}{2\cdot 24}\partial_y^4\big(x'(y) \hat{W}_{0,2}(x(y),x(y))\partial_x(y)^2y\big)+\frac{1}{ 24}\partial_y^3\big(x'(y) \partial_{x(y)}^2\hat{W}_{0,2}(x(y),x'(y'))\vert_{x'(y')=x(y)}\big).
\end{align}
For the third $\Gamma_3$ and fourth graph $\Gamma_4$ in Fig.\ref{fig:gn21}, we will need just the leading order of $S(\hbar u \partial_x)=1+\mathcal{O}(\hbar^2)$ and get
\begin{align}\label{W213}
	\frac{1}{8}\partial_y^3\big(x'(y) \hat{W}_{0,2}(x(y),x(y))\hat{W}_{0,2}(x(y),x(y))\big)
\end{align}
and 
\begin{align}\label{W214}
	-\frac{1}{6}\partial_y^2\big(x'(y) W_{0,3}(x(y),x(y),x(y))\big),
\end{align}
respectively. 

Adding all terms together, the result reads
\begin{align}\label{W21}
	W^{\vee}_{2,1}(y)=\eqref{W211}+\eqref{W212}+\eqref{W213}+\eqref{W214}.
\end{align}
\subsubsection{Verification for $(g,n)=(2,1)$ with the Airy Curve}
To verify the functional relation for $(g,n)=(2,1)$, we consider the example of the Airy curve
\begin{align}\label{Airy}
	x(z)=z^2,\qquad y(z)=z,
\end{align}
where $x,y$ maps from $\mathbb{C}\to\mathbb{C}$. Applying TR \eqref{TR} with definition \eqref{cor}, the first few correlators $W_{g,n}$ are
\begin{align*}
	W_{0,1}(x(z))=&y(z)=z\\
	W_{0,2}(x_1(z_1),x_2(z_2))=&\frac{1}{4z_1z_2(z_1-z_2)^2}\\
	\hat{W}_{0,2}(x(z),x(z))=&\frac{1}{16 z^4}\\
	W_{0,3}(x_1(z_1),x_2(z_2),x_3(z_3))=&-\frac{1}{16z_1^3z_2^3z_3^3}\\
	W_{1,1}(x(z))=&-\frac{1}{32z^5}\\
	W_{1,2}(x_1(z_1),x_2(z_2))=&\frac{5z_1^4+5z_2^4+3z_1^2z_2^2}{128z_1^7z_2^7}\\
	W_{2,1}(x(z))=&-\frac{105}{2048z^{11}}.
\end{align*}

On the other hand, we have for the correlator $W^\vee_{g,n}$ after exchanging the role of $x$ and $y$:
\begin{align*}
	W^\vee_{0,1}(y(z))=&x(z)=z^2\\
	W^\vee_{0,2}(y_1(z_1),y_2(z_2))=&\frac{1}{(z_1-z_2)^2}\\
	\hat{W}^\vee_{0,2}(y_1(z_1),y_2(z_2))=&\frac{1}{(z_1-z_2)^2}-\frac{1}{(y_1(z_1)-y_2(z_2))^2}=0\\
	W^\vee_{g,n}(y_I(z_I))=&0\qquad \forall 2g+n-2>0,
\end{align*}
since $y(z)$ is unramified over $\mathbb{C}$. 

Now, we check \eqref{W21} by inserting the previously computed correlators $W_{g,n}$ for the Airy curve \eqref{Airy}. We write (with some lack of notation) $y_i=y_i(z_i)$ and $x_i(y_i)=x_i(z_i)$, which gives for the derivatives $\frac{\partial}{\partial y_i}=\frac{\partial}{\partial z_i}$ and $\frac{\partial}{\partial x_i}=\frac{\partial}{\partial (z_i)^2}=\frac{1}{2z_i}\frac{\partial}{\partial z_i}$. Furthermore, we need $x'(y)=\frac{dx(z)}{dy(z)}=2z$ and get
\begin{align*}
	\eqref{W211}=\frac{87}{32 z^{10}},\quad \eqref{W212}=-\frac{477}{128 z^{10}},\quad \eqref{W213}=-\frac{63}{128 z^{10}},\qquad \eqref{W214}=\frac{3}{2 z^{10}}.
\end{align*}
All four terms sum up to 0, as it should be, since $W^\vee_{2,1}(y(z))=0$.

Changing the role of $x,y$, we can just interchange $W_{g,n}(x_I(z_I))$ and $W^{\vee}_{g,n}(y_I(z_I))$.  
Inserting it now into \eqref{W21}, we recognise that almost all terms vanish except the last term of \eqref{W211}, which is
\begin{align*}
	&\frac{1}{24\cdot 24\cdot 2}\partial^5_{y(z)}\big(x'(y)(\partial_{x(z)}^2y(z))^2\big)
	=\frac{1}{24\cdot 24\cdot 2}\frac{\partial^5}{\partial (z^2)^5}\bigg(\frac{4}{2z}\bigg)=-\frac{105}{2048z^{11}},
\end{align*}
as expected.

\subsubsection{$(g,n)=(g,1)$}
For $n=1$, we can get a much simpler structure for the graphs appearing in $\G_1$. These graphs can be classified quite explicitly together with their automorphisms.
We will take \eqref{ThmEq} for the case $(g,1)$, reformulate the weight $\hat{c}$, express the symmetry factors and carry out the sum over graphs explicitly.

If we take Theorem \ref{Thm:main} into account and follow similar steps as in \eqref{Phi1} for the weights, we can just replace the argument of $\Phi_n$ in \eqref{Phi1} by $x_i+\varepsilon_i\frac{\hbar u_i}{2}$. This happens once again because of the two different formal derivatives $e^{a\partial_x}f(x)=f(x+a)$ and $e^{ax\partial_x}f(x)=f(xe^a)$.
It is convenient to define
\begin{align*}
	\hat{\Phi}_n(x_1,...,x_n;\hbar,u_I):=&\hat{c}(u_I,x_I)=\bigg(\prod_{i= 1}^n\hbar u_i S(\hbar u_i\partial_{x_i})\bigg)\big( W_{n}(x_1,...,x_n)\big)\\
	=&\sum_{(\varepsilon_1,...,\varepsilon_n)\in\{ 1,-1\}^n} (-1)^{\#(\varepsilon_i=-1)} \Phi_n\bigg(x_1+\varepsilon_1\frac{\hbar u_1}{2},...,x_n+\varepsilon_n\frac{\hbar u_n}{2}\bigg),
\end{align*}
where the sum is over all different configurations of $(\varepsilon_1,...,\varepsilon_n)$ with $\varepsilon_i\in \{1,-1\}$, which are $2^n$.
The $\Phi_n(x_1,...,x_n):=\int_{o}^{x_1} dx'_1 ...\int_o^{x_n} dx'_n W_{n}(x'_1,...,x'_n)$ are the primitives of the correlators.

Then, for any $\Gamma\in \G_1$, the automorphisms factorise in two different sets of permutations. First, an overall permutation which permutes the $\bullet$-vertices, and second the permutation of the edges corresponding to one $\bullet$-vertex. This means that for $\mathcal{I}(\Gamma)$ with $\Gamma\in \G_1$, we find
\begin{align*}
	|\mathrm{Aut}(\Gamma)|=|\I(\Gamma)|\prod_{I\in \I(\Gamma)}|I|!,
\end{align*}
where $|\I(\Gamma)|$ is the cardinality of $\I(\Gamma)$ meaning the number of $\bullet$-vertices, and $|I|$ the cardinality of $I$ meaning the number of edges attached to the $\bullet$-vertex characterised by $I$. Since we are summing over all $\Gamma\in \G_1$, the sum together with the symmetry factor $\frac{1}{|\mathrm{Aut}(\Gamma)|}$ can be combined in an exponential-function (similar to the $1$-point correlators $W_1$ included in $\hat{O}$). The symmetry factor of $\frac{1}{|\I(\Gamma)|}$ generates the coefficient of the exponential. Thus, \eqref{ThmEq} can be written for $n=1$ as
\begin{align*}
	W^{\vee}_{g,1}(y) =[\hbar^{2g-1}]\sum_{m\geq 0}(-\partial_y)^m[u^m](-x'(y))\frac{\exp\bigg(\sum_{i=1}^\infty \frac{\hat{\Phi}_i(x(y),...,x(y);\hbar,u)}{i!}-\delta_{i,1}y\bigg)}{u\hbar}.
\end{align*}
Multiplying by $\hbar^{2g-1}$ and summing over $g$ (we let the sum over $m$ start at $-1$ for $g=0$) yields 
\begin{align}\label{W1}
	W^{\vee}_{1}(y)=\sum_{m\geq 0}(-\partial_y)^m[u^m](-x'(y))\frac{\exp\bigg(\sum_{i=1}^\infty \frac{\hat{\Phi}_i(x(y),...,x(y);\hbar,u)}{i!}-\delta_{i,1}y\bigg)}{u\hbar},
\end{align}
where $W^{\vee}_{1}(y)=\sum_{g=0}^\infty \hbar^{2g-1} W^{\vee}_{g,1}(y)$. The expression \eqref{W1} indicates a connection to the wave function of the quantum spectral curves, see \cite{Norbury:2015lcn} for a survey.

\subsection{A further combinatorial interpretation}\label{Sec.comb}
The formula \eqref{ThmEq} of Theorem \ref{Thm:main} and the examples have shown that each graph $\Gamma\in \G_n$ can generate several terms through the exponential included in the operator $\hat{O}(x(y))$ and/or the formal expansion of $S(\hbar u \partial_x)$. We will now provide a new set of graphs, which will capture the entire structure term by term. In other words, each of the graphs defined in Definition \ref{def:graphnew} will generate exactly one term in the functional relation.  
\begin{definition}\label{def:graphnew}
	Let $\mathcal{G}_{n}^{(g)}$ be the set of decorated connected bicoloured  graphs $\Gamma$ with $\bigcirc$-vertices and  $\bullet$-vertices, where the number of $\bigcirc$-vertices is $n$. A graph $\Gamma$ satisfies the following conditions:
	\begin{itemize}
		\item[-] the  $\bigcirc$-vertices are labelled from $1,...,n$
		\item[-] associate an integer $h_j\geq 0$ to any edge
		\item[-] associate an integer $g_i\geq 0$ to any $\bullet$-vertex
		\item[-] edges are only connecting $\bullet$-vertices with $\bigcirc$-vertices
		\item[-] no 1-valent $\bullet$-vertices with $g_j=0$ and adjacent edge with $h=0$ is permitted, where $g_i$ is the associated integer of the $\bullet$-vertex and $h$ the associated integer of the adjacent edge
		\item[-] $g=\sum_ig_i+\sum_jh_j+b_1$, where $b_1$ is the first Betti number of $\Gamma$
	\end{itemize}
	For a graph $\Gamma\in \mathcal{G}^{(g)}_{n}$, let $r_{i}(\Gamma)$ be the valence of the $i^{\text{th}}$ $\bigcirc$-vertex. 
	
	Let $(g_i,D)$ be a couple associated to a $\bullet$-vertex, where $g_i$ is the associated integer of the $\bullet$-vertex and $D=\{(i_1,h_1),...,(i_k,h_k)\}$ is a set of couples, where $i_j\in \{1,...,n\}$ is the label of the $\bigcirc$-vertex connected to the $\bullet$-vertex through the edge with the associated number $h_j$.	
	Let $\mathcal{I}^{g}(\Gamma)$ be the set of all couples $(g_i,D)$ for a graph $\Gamma\in \mathcal{G}_{n}^{(g)}$.
\end{definition}\noindent
The automorphism group $\mathrm{Aut}(\Gamma)$ with $\Gamma\in \G^{(g)}_n$ consists of permutations of edges which preserve the structure of $\Gamma$ considering the labels of the $\bigcirc$-vertices and all associated decorations to edges and $\bullet$-vertices. A graph $\Gamma\in \mathcal{G}^{(g)}_n$ is up to automorphisms completely characterised by $\mathcal{I}^{g}(\Gamma)$. 

The last two conditions of Definition \ref{def:graphnew}
\begin{itemize}
	\item[-] no 1-valent $\bullet$-vertex with $g_j=0$ and adjacent edge with $h=0$ is permitted
	\item[-] $g=\sum_ig_i+\sum_jh_j+b_1$
\end{itemize}
imply that the set $\G^{(g)}_n$ is for any $g$ and $n$ finite. 

The difference to the graphs defined in Definition \ref{def:graph} are that we permit some $\bullet$-vertices of valence 1. Furthermore, we include additionally decorations for the edges and the $\bullet$-vertices by a kind of intrinsic genus $h_j$ and $g_i$ respectively. The 1-valent $\bullet$-vertices are generated by the exponential of $\hat{O}(x(y))$-operator defined in \eqref{OO-Operator}. The decoration of the $\bullet$-vertex with $g_i$ is generated by the expansion (as defined in Theorem \ref{Thm:mainintro}) of any 
\begin{align*}
	W_{n}(x_I)=\sum_{g=0}^\infty\hbar^{2g+n-2}W_{g,n}(x_I)
\end{align*}
and picks up the $\hbar^{2g_i+n-2}$-order. The decoration of the edges with $h_j$ is generated by the expansion of $S(\hbar u \partial_x)$ and picks up the $\hbar^{2h_j}$-order. 

\begin{definition}[weight]\label{def:weightnew}
	Let $(g_i,D)\in \I^g(\Gamma)$ with $\Gamma\in \mathcal{G}^{(g)}_n$ and $D=\{(i_1,h_1),...,(i_k,h_k)\}$ as in Definition \ref{def:graphnew}. Let further be $I_D=\{i_1,...,i_k\}$, then we define the weight $\varpi^{g_i}_D$ associated to $(g_i,D)$ by
	\begin{align*}
		\varpi^{g_i}_D:=\prod_{l=1}^k\bigg( \frac{1}{2^{2h_l}(2h_l+1)!}\frac{\partial^{2h_l}}{\partial x_{i_l}^{2h_l}}\bigg) W_{g_i,k}(x_{I_D})
	\end{align*}
	and for $(0,\{(i,h_1),(i,h_2)\})$ the special case
	\begin{align*}
		\varpi^{0}_{\{(i,h_1),(i,h_2)\}}:=\bigg( \frac{1}{2^{2h_1}(2h_1+1)!}\frac{\partial^{2h_1}}{\partial x_{i}^{2h_1}}\bigg)\bigg( \frac{1}{2^{2h_2}(2h_2+1)!}\frac{\partial^{2h_2}}{\partial x^{2h_2}}\bigg) \bigg(W_{0,2}(x_i,x)-\frac{1}{(x_i-x)^2}\bigg)\bigg\vert_{x=x_i}.
	\end{align*}
\end{definition}\noindent
The functional relation of Theorem \ref{Thm:main} turns into:
\begin{proposition}\label{Prop:newformula}
	For $2g-2+n>0$, the functional relation holds
	\begin{align}\label{ThmEqComb}
		W^\vee_{g,n}(y_1,...,y_n) =\sum_{\Gamma\in\G^{(g)}_n}\frac{1}{|\mathrm{Aut}(\Gamma)|}\prod_{i=1}^n\bigg( -\frac{\partial}{\partial y_i}\bigg)^{r_i(\Gamma)+2H_i(\Gamma)-1}\!\!\!\!\!\!\!\!\! (-x'_i(y_i))\prod_{(g_j,D)\in\I^g(\Gamma)}\varpi^{g_j}_D,
	\end{align}
	where the graphs $\mathcal{G}_n^{(g)}$ are defined in Definition \ref{def:graphnew}, $r_i(\Gamma)$ is the valence of the $i^{\text{th}}$ $\bigcirc$-vertex of the graph $\Gamma\in \G^{(g)}_n$ and $H_i(\Gamma)=\sum_{e\in E_i}h_e$ is the sum over all associated integer of the edges adjacent to the $i^{\text{th}}$ $\bigcirc$-vertex.

	\begin{proof}
	The expansion of the exponential in Theorem \ref{Thm:main} yields for the  $i^{\text{th}}$-labelled $\bigcirc$-vertex a power series with a factor $\frac{1}{k_i!}$ at $k_i^{\text{th}}$ power. This expansion generates graphically $k_i$ 1-valent $\bullet$-vertices attached to the $i^{\text{th}}$ $\bigcirc$-vertex. The symmetry factor becomes 
	\begin{align*}
		\frac{1}{|\mathrm{Aut}(\Gamma)|k_1!....k_n!}=\frac{1}{|\mathrm{Aut}(\tilde{\Gamma})|},
	\end{align*}
	where $\Gamma\in \G_n$ and $\tilde{\Gamma}\in \G^{(g)}_n$\footnote{Edge and $\bullet$-vertex decorations are ignored. Strictly speaking, at this stage, the set $\G^{(g)}_n$ has no decorations, which will be incorporated through the expansions of $S(x)$ and $W_n$.}  for some $g$, since each $k_i$ 1-valent $\bullet$-vertices gives only additional automorphisms by the $k_i!$ permutations for each $i$.  
	
	The expansion of $S(\hbar u\partial_x)$ in $\hat{c}(u_I,x_I)$ of \eqref{ccv} generates at order $[\hbar^{2h}]$ the factor $ \frac{u^{2h}\partial_x^{2h}}{2^{2h}(2h+1)!}$, which associates graphically the genus $h$ (the edge decoration of $\tilde{\Gamma}$) to the corresponding edge. The genus expansion of the correlator in $\hat{c}(u_I,x_I)$ of \eqref{ccv} is 
	\begin{align*}
		W_{|I|}(x_I)=\sum_{g_i=0}^\infty\hbar^{2g_i+|I|-2}W_{g_i,|I|}(x_I)
	\end{align*}
	and associates graphically the genus $g_i$ to the $\bullet$-vertex with a factor $\hbar^{2g_i+|I|-2}$. 
	
	Restricting to the $[\hbar^{2g+n-2}]$ coefficient in \eqref{ThmEq} picks from the different expansion all terms such that the last condition of Definition \ref{def:graphnew} is satisfied $g=\sum_ig_i+\sum_eh_e+b_1$ (using the additivity property of the Euler-characteristic). The factors of $u_i$ from the expansions of the $S(\hbar u_i\partial_{x_i})$ add together to $u_i^{r_i+H_i-1}$, where $H_i(\Gamma)=\sum_{e\in E_i}h_e$ is the sum over all genera associated to the edges adjacent to the $i^{\text{th}}$ $\bigcirc$-vertex and $r_i$ its valence. 
	
	The 1-valent $\bullet$-vertex with adjacent edge associated with $h=0$ does not appear due to the subtraction $-y u$ in the exponential of $\hat{O}(x(y))$, which cancels exactly the first order expansion of $\hbar u S(\hbar u\partial_x)W_1(x(y))$.
	\end{proof}
\end{proposition}

\subsubsection{Example $(g,n)=(1,2)$}
We have considered this example already in \ref{exgn12}. The graphs shown in Fig.\ref{fig:gn12} are also included in $\G^{(1)}_1$ with all edges and $\bullet$-vertices decorated by $0$ (due to the Betti number $b_1=1$), except the left upper graph (since it has Betti number $b_1=0$) which turns in $\G^{(1)}_1$ into the graphs shown in Fig.\ref{fig:gn12new}.
\begin{figure}[h]
	\scalebox{1}{
\begingroup%
  \makeatletter%
  \providecommand\color[2][]{%
    \errmessage{(Inkscape) Color is used for the text in Inkscape, but the package 'color.sty' is not loaded}%
    \renewcommand\color[2][]{}%
  }%
  \providecommand\transparent[1]{%
    \errmessage{(Inkscape) Transparency is used (non-zero) for the text in Inkscape, but the package 'transparent.sty' is not loaded}%
    \renewcommand\transparent[1]{}%
  }%
  \providecommand\rotatebox[2]{#2}%
  \newcommand*\fsize{\dimexpr\f@size pt\relax}%
  \newcommand*\lineheight[1]{\fontsize{\fsize}{#1\fsize}\selectfont}%
  \ifx\svgwidth\undefined%
    \setlength{\unitlength}{371.59836158bp}%
    \ifx\svgscale\undefined%
      \relax%
    \else%
      \setlength{\unitlength}{\unitlength * \real{\svgscale}}%
    \fi%
  \else%
    \setlength{\unitlength}{\svgwidth}%
  \fi%
  \global\let\svgwidth\undefined%
  \global\let\svgscale\undefined%
  \makeatother%
  \begin{picture}(1,0.18305822)%
    \lineheight{1}%
    \setlength\tabcolsep{0pt}%
    \put(0,0){\includegraphics[width=\unitlength,page=1]{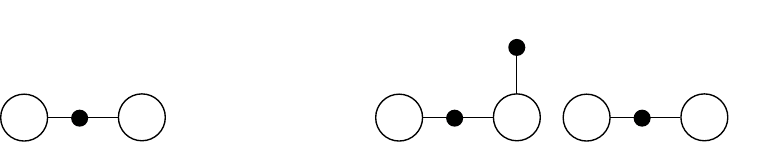}}%
    \put(0.10180028,0.05134259){\makebox(0,0)[t]{\lineheight{1.25}\smash{\begin{tabular}[t]{c}\tiny 1\end{tabular}}}}%
    \put(0.66692645,0.14216646){\makebox(0,0)[t]{\lineheight{1.25}\smash{\begin{tabular}[t]{c}\tiny 1\end{tabular}}}}%
    \put(0.91921512,0.07152567){\makebox(0,0)[t]{\lineheight{1.25}\smash{\begin{tabular}[t]{c}\tiny 1\end{tabular}}}}%
    \put(0,0){\includegraphics[width=\unitlength,page=2]{gn12graphsnew.pdf}}%
    \put(0.37427189,0.04125105){\makebox(0,0)[t]{\lineheight{1.25}\smash{\begin{tabular}[t]{c}\tiny 1\end{tabular}}}}%
  \end{picture}%
\endgroup%
}
	\caption{The upper left graph in Fig.\ref{fig:gn12} from $\mathcal{G}_1$ turns by considering the graphs $\mathcal{G}^{(1)}_1$ into these four graphs (up to possible permutation of the $\bigcirc$-vertex after label them). All edges and $\bullet$-vertices  with no number are associated with a 0.}
	\label{fig:gn12new}
\end{figure}
Applying Definition \ref{def:graphnew} and the formula \eqref{ThmEqComb} with the weight given by Definition \ref{def:weightnew} leads to the graphs shown in Fig.\ref{fig:gn12new}, providing all contributions in \eqref{W121}. The graphs with Betti number $b_1=1$ have exactly the same contribution as before. Consequently, we achieve the same functional relation as in \ref{exgn12} from a more transparent combinatorial point of view.

\subsubsection{Example $(g,n)=(2,1)$}
We have considered the example with $(g,n)=(2,1)$ already in \ref{exgn21}. All graphs with Betti number $b_1=0$ included in $\G^{(2)}_1$ are shown in Fig.\ref{fig:gn21new1}.
\begin{figure}[h]
	\scalebox{1}{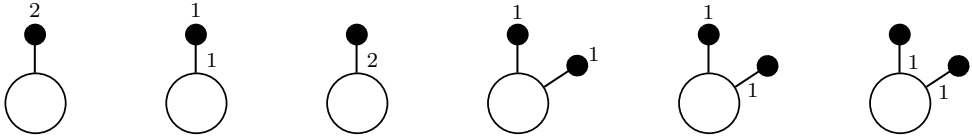}
	\caption{The first graph in Fig.\ref{fig:gn21} turns by considering $\G^{(2)}_1$ into these six graphs, where the last three graphs have $|\mathrm{Aut}(\Gamma)|=2$, but the second last graph is taken two times after permuting the edges.}
	\label{fig:gn21new1}
\end{figure}
We need in total a genus of $g=2$, which means that the associated numbers of edges and $\bullet$-vertices sum up to two. Taking Definition \ref{def:weightnew} for the weights, we deduce exactly \eqref{W211}. 

Just as an example consider the last graph in Fig.\ref{fig:gn21new1}, we find for $H(\Gamma)=2+2(1+1)-1=5$ in formula \eqref{ThmEqComb}. Then, take for the automorphisms $|\mathrm{Aut}(\Gamma)|=2$. Each $\bullet$-vertex with an attached edge associated with a 1 has the weight $\frac{\partial_{x(y)}^2y}{24}$. Thus, the contribution of the last graph in Fig.\ref{fig:gn21new1} is due to \eqref{ThmEqComb} 
\begin{align*}
	\frac{1}{24\cdot 24\cdot 2}\partial^5_y\big(x'(y)(\partial_{x(y)}^2y)^2\big)
\end{align*}
which is indeed the last term in \eqref{W211}.

All graphs with Betti number $b_1=1$, which are included in $\G^{(2)}_1$, are shown in Fig.\ref{fig:gn21new2}. These turn with Definition \ref{def:weightnew} and formula \eqref{ThmEqComb} into \eqref{W212}. 
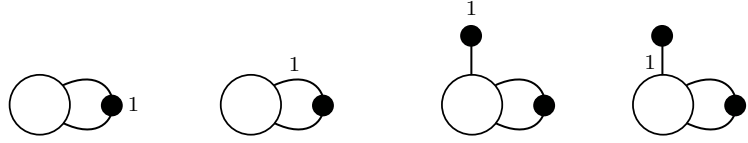
\begin{figure}[h]
	\scalebox{1}{
\begingroup%
  \makeatletter%
  \providecommand\color[2][]{%
    \errmessage{(Inkscape) Color is used for the text in Inkscape, but the package 'color.sty' is not loaded}%
    \renewcommand\color[2][]{}%
  }%
  \providecommand\transparent[1]{%
    \errmessage{(Inkscape) Transparency is used (non-zero) for the text in Inkscape, but the package 'transparent.sty' is not loaded}%
    \renewcommand\transparent[1]{}%
  }%
  \providecommand\rotatebox[2]{#2}%
  \newcommand*\fsize{\dimexpr\f@size pt\relax}%
  \newcommand*\lineheight[1]{\fontsize{\fsize}{#1\fsize}\selectfont}%
  \ifx\svgwidth\undefined%
    \setlength{\unitlength}{275.03708614bp}%
    \ifx\svgscale\undefined%
      \relax%
    \else%
      \setlength{\unitlength}{\unitlength * \real{\svgscale}}%
    \fi%
  \else%
    \setlength{\unitlength}{\svgwidth}%
  \fi%
  \global\let\svgwidth\undefined%
  \global\let\svgscale\undefined%
  \makeatother%
  \begin{picture}(1,0.21957112)%
    \lineheight{1}%
    \setlength\tabcolsep{0pt}%
    \put(0,0){\includegraphics[width=\unitlength,page=1]{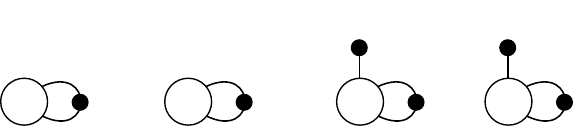}}%
    \put(0.16928971,0.03343145){\makebox(0,0)[t]{\lineheight{1.25}\smash{\begin{tabular}[t]{c}\tiny 1\end{tabular}}}}%
    \put(0.38744212,0.08796955){\makebox(0,0)[t]{\lineheight{1.25}\smash{\begin{tabular}[t]{c}\tiny 1\end{tabular}}}}%
    \put(0.62740975,0.16432289){\makebox(0,0)[t]{\lineheight{1.25}\smash{\begin{tabular}[t]{c}\tiny 1\end{tabular}}}}%
    \put(0.87010433,0.09069646){\makebox(0,0)[t]{\lineheight{1.25}\smash{\begin{tabular}[t]{c}\tiny 1\end{tabular}}}}%
  \end{picture}%
\endgroup%
}
	\caption{The second graph in Fig.\ref{fig:gn21} turns by considering $\G^{(2)}_1$ into these four graphs, where all have $|\mathrm{Aut}(\Gamma)|=2$, but the second  graph is taken two times after permuting the edges.}
	\label{fig:gn21new2}
\end{figure}

Finally, all graphs with Betti number $b_1=2$ are already shown in Fig.\ref{fig:gn21} and give the same contribution as before in \ref{exgn21} such that we derive the same functional relation for $W^\vee_{2,1}$ from Proposition \ref{Prop:newformula}.

\section{Application in Free Probability}\label{Sec.free}
This section aims to apply the $x$-$y$ symplectic transformation to free probability and to give the most simple functional relation between generating series of higher order moments and generating series of higher order free cumulants. 

To understand this, we first recall some facts about the combinatorics of ordinary and fully simple maps. The primary example of TR is the combinatorial problem of ordinary maps. It is carried out gently in the introductory book of TR by Eynard \cite{Eynard:2016yaa}. It can be represented as random $N\times N$-hermitian unitary invariant matrix model with the partition function
\begin{align}
	\mathcal{Z}:=\int_{H_N}dM\,e^{-N\mathrm{Tr}V(M)}
\end{align}
where $V(x)=\sum_{k=1}^{d+1}\frac{t_{k}}{k}x^k$ is a polynomial of degree $d+1$ and $H_N$ the space of hermitian $N\times N$-matrices. This partition function is by definition a formal matrix model in the sense that the expansion of the integrand (except for the Gaussian part) is interchanged with the integration. Let the probability weight be $d\mu:=\frac{1}{\mathcal{Z}}dM\,e^{-N\mathrm{Tr}V(M)}$ and $\langle \mathcal{O}(M)\rangle:=\int_{H_N}\mathcal{O}(M)d\mu$, we define the following correlators by the connected part of the expectation value of the resolvents 
\begin{align}\label{ordinary}
\hat{W}_{g,n}(x_1,...,x_n):=&N^{2g+n-2}\bigg\langle \prod_{i=1}^{n}\mathrm{Tr}\frac{1}{x_i-M}\bigg\rangle_{c}
=N^{2g+n-2}\sum_{k_1,...,k_n\geq 0}\frac{\bigg\langle \prod_{i=1}^{n}\mathrm{Tr}M^{k_i}\bigg\rangle_{c}}{x_1^{k1+1}...x_n^{k_n+1}},
\end{align}
which gives the combinatorics of ordinary maps. Up to some subtleties for $(g,n)\in \{(0,1),(0,2)\}$, $\hat{W}_{g,n}$ corresponds exactly to $W_{g,n}$ considered before with a very explicit spectral curve (see \cite{Eynard:2016yaa} for more details).

For a cycle $\gamma=(c_1c_2...c_k)$ of the symmetric group $S_N$ of length $k$, we denote
\begin{align*}
	\mathcal{P}^{(k)}_\gamma (M):=\prod_{m=1}^k M_{c_m,\gamma(c_m)}.
\end{align*}
Then, we define for pairwise disjoint cycles $\gamma_i\in S_N$ the correlators
\begin{align}\label{fullysimple}
	\hat{W}^\vee_{g,n}(y_1,...,y_n):=&
	=N^{2g+n-2}\sum_{k_1,...,k_n\geq 0}N^{k_1+...+k_n}\bigg\langle \prod_{i=1}^{n}\mathcal{P}^{(k_i)}_{\gamma_i}(M) \bigg\rangle_{c}y_1^{k_1-1}...y_n^{k_n-1}.
\end{align}
This correlators $\hat{W}^\vee_{g,n}$ are the generating series of the so-called fully simple maps
(see \cite{Garcia-Failde:2019iuf} for more details). 

It was then proven independently by two different groups with different techniques that $\hat{W}_{g,n}$ of \eqref{ordinary} and $\hat{W}^\vee_{g,n}$ of \eqref{fullysimple} are related by the $x$-$y$ symplectic transformation in TR \cite{Bychkov:2021hfh,Borot:2021eif}. In other words, $\hat{W}^\vee_{g,n}$ of \eqref{fullysimple} is related by Theorem \ref{Thm:main} (up to some subtleties for $(g,n)=\{(0,1),(0,2)\}$) to $\hat{W}_{g,n}$ of \eqref{ordinary}.

We will now turn to the theory of free probability, which arose from operator algebra from the pioneering works of Voiculescu \cite{Voiculescu1986AdditionOC,Voiculescu1991}. Instead of  talking about independent variables (as in classical probability theory), the notion of freeness  for non-commutative variables was introduced, where
free cumulants play the same role as classical cumulants for usual commutative random variables. In later works \cite{Mingo2004SecondOF,MINGO2007212,Collins2006SecondOF}, this notion was generalised to higher order freeness (we refer to these works for more details). For a random variable $a$, the generating series of higher order moments $\varphi_n[a^{k_1},...,a^{k_n}]$ and higher order free cumulants $\kappa_{k_1,...,k_n}[a,...,a]$ is defined by
\begin{align*}
	M_n(X_1,...,X_n):=&\delta_{n,1}+\sum_{k_1,...,k_n\geq 1}\varphi_n[a^{k_1},...,a^{k_n}]\prod_{i=1}^nX_i^{k_i}\\
	C_n(Y_1,...,Y_n):=&\delta_{n,1}+\sum_{k_1,...,k_n\geq 1}\kappa_{k_1,...,k_n}[a,...,a]\prod_{i=1}^nY_i^{k_i}.
\end{align*}
The functional relation for $C_1$ and $M_1$ was derived in \cite{Voiculescu1986AdditionOC}
\begin{align}\label{voi}
	C_1(X M(X))=M_1(X)
\end{align}
and at second order in \cite{Mingo2004SecondOF}
\begin{align*}
	M_2(X_1,X_2)+\frac{X_1X_2}{(X_1-X_2)^2}=\frac{d\log Y_1}{d\log X_1} \frac{d\log Y_2}{d\log X_2}\bigg(C_2(Y_1,Y_2)+\frac{Y_1Y_2}{(Y_1-Y_2)^2}\bigg).
\end{align*}

A fundamental example of free probability is given if we take the algebra of the random variable $a$ as the algebra of random matrices. The moments and free cumulants are denoted by $\varphi^M_{l_1,...,l_n}$ and $\kappa^M_{l_1,...,l_n}$. For this algebra we have
\begin{align*}
	\varphi^M_{l_1,...,l_n}=&\lim_{N\to \infty} N^{n-2}\bigg\langle \prod_{i=1}^{n}\mathrm{Tr}M^{k_i}\bigg\rangle_{c}\\
	\kappa^M_{l_1,...,l_n}=&\lim_{N\to \infty} N^{n+k_1+...+k_n-2}\bigg\langle \prod_{i=1}^{n}\mathcal{P}^{(k_i)}_{\gamma_i}(M) \bigg\rangle_{c}.
\end{align*}
Recently, the notion of higher order freeness was even generalised to higher genus in \cite{Borot:2021thu} by so-called  surfaced permutations to \textit{surfaced free probability}. They manage to define consistently higher order moments $\varphi_{g,n}(a^{k_1},...,a^{k_n})$ and higher order free cumulants $\kappa_{g;k_1,...,k_n}(a,...,a)$ of higher genus. The higher genus generating functions are defined by
\begin{align}\label{freeM}
M_{g,n}(X_1,...,X_n):=&\delta_{n,1}\delta_{g,0}+\sum_{k_1,...,k_n\geq 1}\varphi_{g,n}[a^{k_1},...,a^{k_n}]\prod_{i=1}^nX_i^{k_i}\\\label{freeC}
C_{g,n}(Y_1,...,Y_n):=&\delta_{n,1}\delta_{g,0}+\sum_{k_1,...,k_n\geq 1}\kappa_{g;k_1,...,k_n}[a,...,a]\prod_{i=1}^nY_i^{k_i}
\end{align}
and follow the functional relation implied by the $x$-$y$ symplectic transformation of TR. Therefore, we are able to give a simpler functional relation which is valid for \eqref{freeM} and \eqref{freeC} defined properly in \cite{Borot:2021thu}.
\begin{corollary}\label{cor:free}
	Let $(g_i,D)\in \I^g(\Gamma)$ with $\Gamma\in \mathcal{G}^{(g)}_n$ and $D=\{(i_1,h_1),...,(i_k,h_k)\}$ as in Definition \ref{def:graphnew}. Let further be $I_D=\{i_1,...,i_k\}$, then we define the weight $\varpi^{g_i,free}_D$ associated to $(g_i,D)$ by
	\begin{align*}
	\varpi^{g_i,free}_D:=\prod_{l=1}^k\bigg( \frac{1}{2^{2h_l}(2h_l+1)!}\frac{\partial^{2h_l}}{\partial Y_{i_l}^{2h_l}}\bigg) \frac{C_{g_i,k}(Y_{I_D})}{\prod_{i\in I_D} Y_i}.
	\end{align*}
	For $2g-2+n>0$ and  $Y_i=X_i M_1(X_i)$, the moment-cumulant relation in free probability for unitarily invariant ensembles reads
	\begin{align}\label{corEqComb}
	&M_{g,n}(X_1,...,X_n)\prod_{i=1}^nX_i \\\nonumber
	=&\sum_{\Gamma\in\G^{(g)}_n}\frac{1}{|\mathrm{Aut}(\Gamma)|}\prod_{i=1}^n\bigg( X_i^2\frac{\partial}{\partial X_i}\bigg)^{r_i(\Gamma)+2H_i(\Gamma)-1} \bigg(X_i^2\frac{dY_i}{dX_i}\bigg)\prod^\prime_{(g_j,D)\in\I^g(\Gamma)}\varpi^{g_j,free}_D,
	\end{align}
	where the graphs $\mathcal{G}_n^{(g)}$ are defined in Definition \ref{def:graphnew}, $r_i(\Gamma)$ is the valence of the $i^{\text{th}}$ $\bigcirc$-vertex of the graph $\Gamma\in \G^{(g)}_n$ and $H_i(\Gamma)=\sum_{e\in E_i}h_e$ is the sum over all associated integer of the edges adjacent to the $i^{\text{th}}$ $\bigcirc$-vertex.
	The primed product $\prod^\prime_{(g_j,D)\in\I^g(\Gamma)}$ replaces $C_{0,2}(Y_i,Y_j)$ by $C_{0,2}(Y_i,Y_j)+\frac{Y_iY_j}{(Y_i-Y_j)^2}$ except for $Y_i=Y_j$.
	\begin{proof}
		Consider the fact that fully simple maps and ordinary maps are related through the $x$-$y$ symplectic transformation \cite{Bychkov:2021hfh,Borot:2021eif}. The relation of moments to ordinary maps and free cumulants to fully simple maps yields the identification
		\begin{align*}
			\frac{W^\vee_{g,n}(\frac{1}{X_1},...,\frac{1}{X_n})}{X_1\cdot ...\cdot X_n}=&M_{g,n}(X_1,...,X_n)+\frac{\delta_{2,n}\delta_{g,0}X_1X_2}{(X_1-X_2)^2}\\
			W_{g,n}(Y_1,...,Y_n)=&\frac{C_{g,n}(Y_1,...,Y_n)}{Y_1\cdot ...\cdot Y_m}+\frac{\delta_{2,n}\delta_{g,0}}{(Y_1-Y_2)^2}.
		\end{align*}
		 Then, the assertion is equivalent to Proposition \ref{Prop:newformula}.
	\end{proof}
\end{corollary}
The moment-cumulant relation in free probability in Corollary \ref{cor:free} becomes much more canonical by considering the shifted generating series
\begin{align*}
	\tilde{M}_{g,n}(X_1,...,X_n):=&\delta_{n,1}\delta_{g,0}X_1+\sum_{k_1,...,k_n\geq 1}\varphi_{g,n}[a^{k_1},...,a^{k_n}]\prod_{i=1}^nX_i^{k_i+1}\\
	\tilde{C}_{g,n}(Y_1,...,Y_n):=&\frac{\delta_{n,1}\delta_{g,0}}{Y_1}+\sum_{k_1,...,k_n\geq 1}\kappa_{g;k_1,...,k_n}[a,...,a]\prod_{i=1}^nY_i^{k_i-1},
\end{align*}
which turns \eqref{voi} into
\begin{align*}
	\tilde{C}_1(\tilde{M}_1(X))=\frac{1}{X}.
\end{align*}

\appendix
\section{Proof of Lemma \ref{lem:tech}}\label{AppA}
Shift the arbitrary function $f\bigg(\frac{x(y)}{\cosh(u/2)}\bigg)\to f\big(x(y)\big)$ (the proof would also work for $f\bigg(\frac{x(y)}{\cosh(u/2)}\bigg)$ but it is just more convenient this way), then the Lemma reads: 
\begin{lemma}
	Let $\Phi_{0,1}(x)=\int_o^x dx'\,y(x')$, $f(x)$ some smooth function and $S(u)=\frac{e^{u/2}-e^{-u/2}}{u}$. Then, the following simplification holds as a formal power series in $u^{2}$
	\begin{align}\nonumber
	&\cosh(u/2)S(u)\sum_{j\geq0}(-\partial_y)^j\bigg[x'(S(u)y)
	f\big(\cosh(u/2) x(S(u)y)\big)\\\label{LemmEq}
	&\qquad\qquad\qquad\qquad\qquad \times \frac{\bigg(\frac{\Phi_{0,1}(e^{u/2} x(S(u) y))-\Phi_{0,1}(e^{-u/2} x(S(u) y))}{u S(u)x(S(u)y)}-y\bigg)^j}{j!}\bigg]\\\nonumber
	=&\sum_{j\geq0}(-\partial_y)^j\bigg[x'(y)
	f\big(x(y)\big)\frac{\bigg(\frac{\Phi_{0,1}(\frac{e^{u/2} x( y)}{\cosh(u/2)})-\Phi_{0,1}(\frac{e^{-u/2} x( y)}{\cosh(u/2)})}{u S(u)\frac{x(y)}{\cosh(u/2)}}-y\bigg)^j}{j!}\bigg].
	\end{align} 
	This is a simplification in the sense that the rhs does not contain any explicit powers of $y$, whereas the lhs does through the expansion of $x(S(u)y)$.
\end{lemma}
\begin{proof}
	The Lemma will be proven by induction, where we take on both sides the derivative wrt $u$ and represent it as derivatives wrt $y$. This will then be equivalent to the induction hypothesis.
	
	Note that both expressions in the brackets with power $j$ are at least of order $u^2$, which restricts both sums over $j$ up to $n$, if take this equation at order $u^{2n}$. Furthermore, since $f$ is an arbitrary function, it can have a $u$-expansion by itself
	\begin{align*}
		f(x(y))=\sum_{k\geq 0}f_k(x(y))u^k.
	\end{align*}
	This implies by linearity and by the arbitrariness of $f(x(y))$ that the Lemma at order $u^{2n}$ is equivalent to 
	\begin{align}\nonumber
	&\sum_{j=0}^{n-k}(-\partial_y)^j\bigg[x'(S(u)y)\cosh(u/2)S(u)
	f_k\big(\cosh(u/2) x(S(u)y)\big)\\\label{genLemmEq}
	&\qquad \times\frac{\bigg(\frac{\Phi_{0,1}(e^{u/2} x(S(u) y))-\Phi_{0,1}(e^{-u/2} x(S(u) y))}{u S(u)x(S(u)y)}-y\bigg)^j}{j!}\bigg]\\\nonumber
	=&\sum_{j=0}^{n-k}(-\partial_y)^j\bigg[x'(y)
	f_k\big(x(y)\big)\frac{\bigg(\frac{\Phi_{0,1}(\frac{e^{u/2} x( y)}{\cosh(u/2)})-\Phi_{0,1}(\frac{e^{-u/2} x( y)}{\cosh(u/2)})}{u S(u)\frac{x(y)}{\cosh(u/2)}}-y\bigg)^j}{j!}\bigg]
	\end{align} 
	taken at order $u^{2n-k}$ for all $k$ between $0$ and $2n$.
	
	We proceed by induction in $n$, where we are looking at the order $u^{2n}$. The initial case with $n=0$ holds obviously.
	
	Assume now without loss of generality that $f(x(y))$ is independent of $u$. Then, we take the derivative wrt $u$ of the rhs of \eqref{LemmEq} taken at order $u^{2n}$
	\begin{align*}
		&-\partial_y\sum_{j=1}^n(-\partial_y)^{j-1}\bigg[x'(y)
		f\big(x(y)\big)\frac{\bigg(\frac{\Phi_{0,1}(\frac{e^{u/2} x( y)}{\cosh(u/2)})-\Phi_{0,1}(\frac{e^{-u/2} x( y)}{\cosh(u/2)})}{u S(u)\frac{x(y)}{\cosh(u/2)}}-y\bigg)^{j-1}}{(j-1)!}\\
		&\times \bigg(\partial_t\frac{\Phi_{0,1}(\frac{e^{t/2} x( y)}{\cosh(u/2)})-\Phi_{0,1}(\frac{e^{-t/2} x( y)}{\cosh(u/2)})}{t S(t)\frac{x(y)}{\cosh(u/2)}}\bigg\vert_{t=u}+\partial_t\frac{\Phi_{0,1}(\frac{e^{u/2} x( y)}{\cosh(t/2)})-\Phi_{0,1}(\frac{e^{-u/2} x( y)}{\cosh(t/2)})}{u S(u)\frac{x(y)}{\cosh(t/2)}}\bigg\vert_{t=u}\bigg)\bigg]\\
		=&-\partial_y\sum_{j=0}^{n-1}(-\partial_y)^{j}\bigg[x'(y)
		\tilde{f}\big(x(y)\big)\frac{\bigg(\frac{\Phi_{0,1}(\frac{e^{u/2} x( y)}{\cosh(u/2)})-\Phi_{0,1}(\frac{e^{-u/2} x( y)}{\cosh(u/2)})}{u S(u)\frac{x(y)}{\cosh(u/2)}}-y\bigg)^{j}}{j!}\bigg],
	\end{align*}
	where we have defined
	\begin{align*}
		\tilde{f}\big(x(y)\big):=&f\big(x(y)\big)\bigg(\partial_t\frac{\Phi_{0,1}(\frac{e^{t/2} x( y)}{\cosh(u/2)})-\Phi_{0,1}(\frac{e^{-t/2} x( y)}{\cosh(u/2)})}{t S(t)\frac{x(y)}{\cosh(u/2)}}\bigg\vert_{t=u}+\partial_t\frac{\Phi_{0,1}(\frac{e^{u/2} x( y)}{\cosh(t/2)})-\Phi_{0,1}(\frac{e^{-u/2} x( y)}{\cosh(t/2)})}{u S(u)\frac{x(y)}{\cosh(t/2)}}\bigg\vert_{t=u}\bigg)\\
		=&\sum_{k=1}^\infty \tilde{f}_k(x(y))u^k,
	\end{align*}
	which has a $u$ expansion. Important here, the function $\tilde{f}$ has no constant term, it starts with $u^1$. Therefore, it has a redundant factor of $u$ which can be removed $\tilde{f}=\frac{\hat{f}}{u}$.
	 Using the expansion of $\hat{f}$ in $u$, we have by induction hypothesis that for each $\hat{f}_k$ \eqref{genLemmEq} holds (replace $f_k$ by $\hat{f}_k$  and $n$ by $n-1$ in \eqref{genLemmEq}).
	
	Next, we compute the $u$-derivative of the lhs of \eqref{LemmEq} at order $u^{2n}$
	\begin{align}\nonumber
		&\sum_{j=1}^n(-\partial_y)^{j}\bigg[x'(S(u)y)\cosh(u/2)S(u)
		f\big(\cosh(u/2) x(S(u)y)\big)\\\nonumber
		&\qquad\times \frac{\bigg(\frac{\Phi_{0,1}(e^{u/2} x(S(u) y))-\Phi_{0,1}(e^{-u/2} x(S(u) y))}{u S(u)x(S(u)y)}-y\bigg)^{j-1}}{(j-1)!}\\\nonumber
		&\qquad\qquad\times \bigg(\partial_t\frac{\Phi_{0,1}(e^{t/2} x(S(u) y))-\Phi_{0,1}(e^{-t/2} x(S(u) y))}{t S(t)x(S(u)y)}\bigg\vert_{t=u}\\\nonumber
		&\qquad\qquad\qquad +\partial_t\frac{\Phi_{0,1}(e^{u/2} x(S(t) y))-\Phi_{0,1}(e^{-u/2} x(S(t) y))}{u S(u)x(S(t)y)}\bigg\vert_{t=u}\bigg)\bigg]\\\nonumber
		&+\sum_{j=0}^n(-\partial_y)^{j}\frac{\bigg(\frac{\Phi_{0,1}(e^{u/2} x(S(u) y))-\Phi_{0,1}(e^{-u/2} x(S(u) y))}{u S(u)x(S(u)y)}-y\bigg)^{j}}{j!}\\\nonumber
		&\qquad \times \partial_t\cosh(t/2)S(t)x'(S(t)y)
		f\big(\cosh(t/2) x(S(t)y)\big)\big\vert_{t=u}\\\label{eq}
		=&-\partial_y\sum_{j=0}^{n-1}(-\partial_y)^{j}\bigg[x'(S(u)y)\cosh(u/2)S(u)
		\tilde{f}\big(\cosh(u/2)x(S(u)y)\big)\\\nonumber
		&\qquad\qquad  \times\frac{\bigg(\frac{\Phi_{0,1}(e^{u/2} x(S(u) y))-\Phi_{0,1}(e^{-u/2} x(S(u) y))}{u S(u)x(S(u)y)}-y\bigg)^{j}}{j!}\bigg]\\\nonumber
		&+\sum_{j=1}^n(-\partial_y)^{j}\bigg[x'(S(u)y)\cosh(u/2)S(u)
		f\big(\cosh(u/2) x(S(u)y)\big)\\\nonumber
		&\qquad\times \frac{\bigg(\frac{\Phi_{0,1}(e^{u/2} x(S(u) y))-\Phi_{0,1}(e^{-u/2} x(S(u) y))}{u S(u)x(S(u)y)}-y\bigg)^{j-1}}{(j-1)!}\\\label{eq1}
		&\qquad\qquad\times \bigg(\partial_t\frac{\Phi_{0,1}(e^{u/2} x(S(t) y))-\Phi_{0,1}(e^{-u/2} x(S(t) y))}{u S(u)x(S(t)y)}\bigg\vert_{t=u}\\\label{eq2}
		&\qquad\qquad\qquad -\partial_t\frac{\Phi_{0,1}(\frac{e^{u/2}\cosh(u/2) x(S(u) y)}{\cosh(t/2)})-\Phi_{0,1}(\frac{e^{-u/2} \cosh(u/2)x(S(u) y)}{\cosh(t/2)})}{u S(u)\frac{\cosh(u/2)x(S(u)y)}{\cosh(t/2)}}\bigg\vert_{t=u}\bigg)\bigg]\\\nonumber
		&+\sum_{j=0}^n(-\partial_y)^{j}\frac{\bigg(\frac{\Phi_{0,1}(e^{u/2} x(S(u) y))-\Phi_{0,1}(e^{-u/2} x(S(u) y))}{u S(u)x(S(u)y)}-y\bigg)^{j}}{j!}\\
		&\qquad \times \partial_t\cosh(t/2)S(t)x'(S(t)y)
		f\big(\cosh(t/2) x(S(t)y)\big)\big\vert_{t=u}\label{eq3},
	\end{align} 
	where the line \eqref{eq1} is added to reconstruct $\tilde{f}(\cosh(u/2) x(S(u)y))$ in line \eqref{eq}. Thus, the induction hypothesis \eqref{genLemmEq} is satisfied for $n-1$ if the last six lines cancel, which will be checked now. 
	
	We substitute in line \eqref{eq1} the derivative wrt $t$ by
	\begin{align*}
		\text{in \eqref{eq1}:}\qquad \partial_t =\frac{S(u)'}{S(u)}y\partial_y.
	\end{align*}
	In line \eqref{eq2}, we substitute the $t$-derivate by
	\begin{align*}
		\text{in \eqref{eq2}:}\qquad \partial_t =-\frac{x(S(u)y)}{S(u)x'(S(u)y)}\frac{\cosh(u/2)'}{\cosh(u/2)}\partial_y.
	\end{align*}
	We also rewrite line \eqref{eq3} as
	\begin{align*}
		&\partial_t\cosh(t/2)S(t)x'(S(t)y)
		f\big(\cosh(t/2) x(S(t)y)\big)\big\vert_{t=u}\\
		=&S(u)'x'(S(u)y)\cosh(u/2)f\big(\cosh(u/2) x(S(u)y)\big)\\
		&+\cosh(u/2)'\partial_y \big[x(S(u)y)f\big(\cosh(u/2) x(S(u)y)\big)\big]\\
		&+\frac{S(u)'}{S(u)}y\partial_y[\cosh(u/2)S(u)x'(S(u)y)f\big(\cosh(u/2) x(S(u)y)\big)]
	\end{align*}
	where $S(u)'=\partial_uS(u)$ and $\cosh(u/2)'=\partial_u\cosh(u/2)$.
	
	Inserting these new representations of derivatives wrt $y$, we find for the last six lines  (\ref{eq1}-\ref{eq3})
	\begin{align}\label{A1}
		&\sum_{j=0}^n(-\partial_y)^{j}S(u)'\cosh(u/2)\bigg(\tilde{y}\partial_{\tilde{y}}\bigg[x'(S(u)\tilde{y})
		f\big(\cosh(u/2) x(S(u)\tilde{y})\big)\\\nonumber
		&\qquad\qquad\qquad\times \frac{\bigg(\frac{\Phi_{0,1}(e^{u/2} x(S(u) \tilde{y}))-\Phi_{0,1}(e^{-u/2} x(S(u) \tilde{y}))}{u S(u)x(S(u)\tilde{y})}-y\bigg)^{j}}{j!}\bigg]\bigg\vert_{\tilde{y}=y}\bigg)\\\nonumber
		&+\sum_{j=0}^n(-\partial_y)^{j}\cosh(u/2)'\bigg(\partial_{\tilde{y}}\bigg[x(S(u)\tilde{y})
		f\big(\cosh(u/2) x(S(u)\tilde{y})\big)\\\nonumber
		&\qquad\qquad\qquad\times \frac{\bigg(\frac{\Phi_{0,1}(e^{u/2} x(S(u) \tilde{y}))-\Phi_{0,1}(e^{-u/2} x(S(u) \tilde{y}))}{u S(u)x(S(u)\tilde{y})}-y\bigg)^{j}}{j!}\bigg]\bigg\vert_{\tilde{y}=y}\bigg)\\\label{eq4}
		&+\sum_{j=0}^n(-\partial_y)^{j}S(u)'\cosh(u/2)\bigg[x'(S(u)y)
		f\big(\cosh(u/2) x(S(u)y)\big)\\\nonumber
		&\qquad\qquad\qquad\times \frac{\bigg(\frac{\Phi_{0,1}(e^{u/2} x(S(u) y))-\Phi_{0,1}(e^{-u/2} x(S(u) y))}{u S(u)x(S(u)y)}-y\bigg)^{j}}{j!}\bigg].
	\end{align}
	Replacing the derivative $\tilde{y}\partial_{\tilde{y}}$ by $\partial_{\tilde{y}} \tilde{y}-1$ in line \eqref{A1} cancels line \eqref{eq4}. The rest can be summarised as
	\begin{align*}
		=&\sum_{j=0}^n(-\partial_y)^{j}\bigg(\partial_{\tilde{y}}\bigg[\partial_u\{x(S(u)\tilde{y})\cosh(u/2)\}
		f\big(\cosh(u/2) x(S(u)\tilde{y})\big)\\\nonumber
		&\qquad\qquad\qquad\times \frac{\bigg(\frac{\Phi_{0,1}(e^{u/2} x(S(u) \tilde{y}))-\Phi_{0,1}(e^{-u/2} x(S(u) \tilde{y}))}{u S(u)x(S(u)\tilde{y})}-y\bigg)^{j}}{j!}\bigg]\bigg\vert_{\tilde{y}=y}\bigg).
	\end{align*}
	Note that the derivative $\partial_{\tilde{y}}$ does not act on $-y$ inside the brackets. Adding this possible action and subtracting it yields
	\begin{align*}
		=&-\sum_{j=0}^n(-\partial_y)^{j+1}\bigg[\partial_u\{x(S(u)y)\cosh(u/2)\}
		f\big(\cosh(u/2) x(S(u)y)\big)\\\nonumber
		&\qquad\qquad\qquad\times \frac{\bigg(\frac{\Phi_{0,1}(e^{u/2} x(S(u) y))-\Phi_{0,1}(e^{-u/2} x(S(u) y))}{u S(u)x(S(u)y)}-y\bigg)^{j}}{j!}\bigg]\\
		&+\sum_{j=1}^n(-\partial_y)^{j}\bigg[\partial_u\{x(S(u)y)\cosh(u/2)\}
		f\big(\cosh(u/2) x(S(u)y)\big)\\\nonumber
		&\qquad\qquad\qquad\times \frac{\bigg(\frac{\Phi_{0,1}(e^{u/2} x(S(u) y))-\Phi_{0,1}(e^{-u/2} x(S(u) y))}{u S(u)x(S(u)y)}-y\bigg)^{j-1}}{(j-1)!}\bigg]\\
		=&0.
	\end{align*}
	The expression vanishes after shifting the second sum $j\to j+1$ for all $j=1,...,n$. One term remains, which is $j=n$ of the first sum. Since we are looking at this expression at order $u^{2n}$ (which is after $u$-derivative $u^{2n-1}$) and we have
	\begin{align*}
		\bigg(\frac{\Phi_{0,1}(e^{u/2} x(S(u) y))-\Phi_{0,1}(e^{-u/2} x(S(u) y))}{u S(u)x(S(u)y)}-y\bigg)^{n}\in \mathcal{O}(u^{2n}),
	\end{align*}
	the remaining terms in $u$ are expanded to the constant term. We find for those
	\begin{align*}
		&\partial_u\{x(S(u)y)\cosh(u/2)\}
		f\big(\cosh(u/2) x(S(u)y)\big)
		\to \partial_u\{x(y)\}
		f\big( x(y)\big)=0,
	\end{align*}
	due to the $u$-derivative.
	
	We can summarise all step performed in the proof as follows: Writing the lhs and rhs of \eqref{genLemmEq} as $f^n_{lhs}(u)$ and  $f^n_{rhs}(u)$, the aim is to prove that $[u^{2n}]f^n_{lhs}(u)=[u^{2n}]f^n_{rhs}(u)$. Differentiating both
	$f^n_{lhs}(u)$ and  $f^n_{rhs}(u)$ with respect to $u$, and since $[u^{2n}]f^n_{lhs}(u)=\frac{1}{2n}[u^{2n-1}]\partial_u f^n_{lhs}(u)$, it amounts to showing that $[u^{2n-1}]\partial_u f^n_{lhs}(u)=[u^{2n-1}]\partial_u f^n_{rhs}(u)$. Then, the following partial differential equation was shown $\partial_uf^n_{lhs}(u)=-\partial_y \tilde{f}^{n-1}_{lhs} (u)$, and similarly for the rhs, for another function $\tilde{f}$. Since the function $\tilde{f}$ has no constant term, it
	has a redundant factor of $u$ which can be removed $\tilde{f}=\hat{f}/u$, so that the induction hypothesis can indeed be used at order $n-1$: $[u^{2n-2}]\hat{f}^{n-1}_{lhs}(u)=[u^{2n-2}]\hat{f}^{n-1}_{rhs}(u)$. This proves by induction in powers of $[u^{2n}]$ that the lemma holds.


\end{proof}

\newcommand{\etalchar}[1]{$^{#1}$}

\end{document}